\documentclass[conference]{IEEEtran}
\IEEEoverridecommandlockouts
\usepackage{cite}
\usepackage{amsmath,amssymb,amsfonts}
\usepackage{booktabs}
\usepackage{graphicx}
\usepackage{textcomp}
\usepackage{xcolor}
\usepackage{algorithm}
\usepackage{algpseudocode}
\usepackage{float}
\usepackage{placeins}
\usepackage{enumitem}
\usepackage{multirow}
\usepackage{graphicx}
\usepackage{comment}
\usepackage{array} 
\usepackage{url} 
\usepackage{orcidlink}
\usepackage{needspace}
\usepackage{multicol}
\usepackage{caption}
\usepackage{amsmath, amsthm, amssymb}
\newtheorem{theorem}{Theorem}
\newtheorem{lemma}{Lemma}

\def\BibTeX{{\rm B\kern-.05em{\sc i\kern-.025em b}\kern-.08em
    T\kern-.1667em\lower.7ex\hbox{E}\kern-.125emX}}

\begin{document}
\pagestyle{plain}

\title{Advancing Practical Homomorphic Encryption for Federated Learning: Theoretical Guarantees and Efficiency Optimizations}

\author{
    Ren-Yi Huang\IEEEauthorrefmark{1}\orcidlink{0000-0002-1245-xxxx},
    Dumindu Samaraweera, \textit{Member, IEEE}\IEEEauthorrefmark{2}\orcidlink{0000-0003-4097-5585},
    Prashant Shekhar\IEEEauthorrefmark{2}\orcidlink{0000-0002-1245-xxxx},
    \\and J. Morris Chang, \textit{Senior Member, IEEE}\IEEEauthorrefmark{1}\orcidlink{0000-0002-1245-xxxx}

    \thanks{\IEEEauthorrefmark{1}Department of Electrical Engineering, University of South Florida, Tampa, FL (emails: hr219@usf.edu, chang5@usf.edu).}
    \thanks{\IEEEauthorrefmark{2}Department of Mathematics, Embry-Riddle Aeronautical University, Daytona Beach, FL (emails: samarawg@erau.edu, shekharp@erau.edu).}  
}
\maketitle

\begin{abstract}
Federated Learning (FL) enables collaborative model training while preserving data privacy by keeping raw data locally stored on client devices, preventing access from other clients or the central server. However, recent studies reveal that sharing model gradients creates vulnerability to Model Inversion Attacks, particularly Deep Leakage from Gradients (DLG), which reconstructs private training data from shared gradients. While Homomorphic Encryption has been proposed as a promising defense mechanism to protect gradient privacy, fully encrypting all model gradients incurs high computational overhead. Selective encryption approaches aim to balance privacy protection with computational efficiency by encrypting only specific gradient components. However, the existing literature largely overlooks a theoretical exploration of the spectral behavior of encrypted versus unencrypted parameters, relying instead primarily on empirical evaluations. To address this gap, this paper presents a framework for theoretical analysis of the underlying principles of selective encryption as a defense against model inversion attacks. We then provide a comprehensive empirical study that identifies and quantifies the critical factors, such as model complexity, encryption ratios, and exposed gradients, that influence defense effectiveness. Our theoretical framework clarifies the relationship between gradient selection and privacy preservation, while our experimental evaluation demonstrates how these factors shape the robustness of defenses against model inversion attacks. Collectively, these contributions advance the understanding of selective encryption mechanisms and offer principled guidance for designing efficient, scalable, privacy-preserving federated learning systems.
\end{abstract}

\begin{IEEEkeywords}
Privacy-preserving Federated Learning, Homomorphic Encryption, Security, Privacy
\end{IEEEkeywords}

%===================================================================
\section{Introduction} \label{introduction}
Federated Learning (FL) is a privacy-preserving paradigm that enables distributed model training while ensuring that clients’ local data remains private from both other clients and the central server. However, FL is still vulnerable to model inversion attacks, such as Deep Leakage from Gradients (DLG)~\cite{b1,b2,b3}, where an adversary can reconstruct private training data by analyzing shared gradients. This presents a significant threat, particularly when a malicious server or attacker exploits gradient updates to recover sensitive client information.

Over the years, several privacy-preserving techniques have been investigated for federated learning across different application settings. These include Differential Privacy (DP), which protects against inference attacks by injecting noise (additive noise) into gradients or model updates; Secure Multi-Party Computation (SMPC), which enables collaborative computation without exposing individual inputs; and Homomorphic Encryption (HE), which permits computations directly on encrypted data while retaining full data utility. Among these approaches, HE is particularly notable for its ability to preserve complete data utility, allowing encrypted model aggregation without revealing raw data and making it a compelling choice for privacy-enhancing mechanisms in FL. Hence, it ensures that even if the server is compromised, model parameters remain protected. However, HE introduces substantial computational and communication overhead, rendering it impractical for many large-scale implementations. To address this challenge, recent research~\cite{b4,b5} has explored selective encryption, encrypting only a subset of the most important model parameters while leaving the rest unencrypted.

While selective encryption shows promise, a key challenge lies in effectively identifying the critical factors and securing the most sensitive model parameters, while balancing the additional computational cost without compromising privacy. Furthermore, analyzing the behavior of encrypted and unencrypted    gradients is also crucial for understanding potential vulnerabilities, as gradient behaviors may reveal patterns that adversaries could exploit in DLG-based reconstruction attacks. Even with selective encryption, none of the existing methods offers a truly efficient, scalable, and practical HE-based solution, especially for complex deep learning models. Without a deeper analysis of both the behavior of encrypted and unencrypted gradients, these methods may still leave FL models susceptible to leakage, making it imperative to develop a more robust and computationally feasible encryption strategy. 

In this work, our goal is to propose a comprehensive mathematical framework to analyze the key parameters governing selective encryption and to identify the minimum encryption ratio that ensures an equivalent level of security, while avoiding unnecessary computational and communication overhead. To achieve this goal, our contributions are as follows: 
\begin{enumerate}
    \item We develop a novel theoretical framework to analyze the effectiveness of selective encryption defenses against reconstruction attacks, enabling scalability to larger deep learning architectures where empirical evaluation is challenging.
    \item Using this theoretical framework, we identify key factors, such as encryption ratio, model complexity, and gradient exposure, that influence the performance of model inversion attacks under selective encryption mechanisms.
    \item We evaluate the accuracy and performance of the proposed framework through both empirical experiments and theoretical analysis, providing comprehensive validation.
\end{enumerate}

The remainder of this paper is organized as follows. Section \ref{Background} provides essential background on federated learning, model inversion attacks, and homomorphic encryption. Section \ref{RelatedWork} reviews existing approaches for mitigating  model inversion attacks, highlighting the current state of the art and identifying key limitations that motivate our work. Section \ref{preliminaries} introduces the workflow of our theoretical analysis of defense mechanisms against model inversion attacks, while Section \ref{intro_BCRLB} presents our proposed mathematical framework, applying the Bayesian Cramér-Rao Lower Bound (BCRLB) to the analysis. Section \ref{derive_bound} derives the theoretical reconstruction lower bound and investigates the factors influencing defense effectiveness. Section \ref{utility_reduction} explores the impact of additive noise on model training in federated learning systems. Section \ref{Experiments} provides comprehensive experimental validation of our theoretical findings, and finally, Section \ref{Conclusion} summarizes our contributions and outlines directions for future research.

\section{Security Vulnerabilities in FL} \label{Background}
In this section, we outline key concepts central to this study: federated learning, attacks targeting FL systems, and the use of homomorphic encryption to protect against these attacks. We begin by introducing FL, followed by a discussion of attacks, with a focus on model inversion attacks, a privacy threat where adversaries attempt to reconstruct sensitive client data from trained models. Finally, we explore homomorphic encryption as a promising approach for safeguarding client privacy in FL environments.

\subsection{Federated Learning} 
Federated learning ~\cite{b6} is a decentralized machine learning approach that allows clients to collaboratively build a global model without revealing their local data. This approach addresses privacy concerns in traditional centralized learning methods by keeping sensitive data localized while still benefiting from diverse datasets across multiple clients. The FL process typically consists of the following iterative steps: 

\begin{enumerate}
  \item \textbf{Client Selection}: The central server selects a subset of clients available to participate in the current training round.
  \item \textbf{Model Distribution}: The server distributes the current global model to the selected clients. This model serves as the starting point for local training in each round.
  \item \textbf{Local Training}: Each selected client trains the received model on their local dataset. This step involves performing one or more epochs of training using techniques such as stochastic gradient descent. 
  \item \textbf{Update Transmission}: Clients send their locally updated model back to the central server. 
  \item \textbf{Aggregation and Model Update}: The server aggregates the received local models and updates the global model with aggregated results.
\end{enumerate}

\subsection{Attacks Targeting FL Systems}
Despite the privacy advantages of federated learning, the transmission of model parameters between the server and clients creates potential vulnerabilities that attackers can exploit, which can be categorized into two types: those that target model security and those that target model privacy ~\cite{b7}. Security attacks include data poisoning, where adversaries alter training data or labels to manipulate the model~\cite{b8,b9,b10}, and free-rider attacks, where participants benefit from the model without contributing~\cite{b11}. Privacy attacks aim to extract sensitive information from the model, such as model inversion attack~\cite{b1,b2,b3}, where adversaries reconstruct local data from model updates. Inference attacks, such as membership or category inference, can also reveal private information, emphasizing the need for stronger privacy protections in FL systems~\cite{b12, b13}.

\begin{figure}[H]
    \centering
    \includegraphics[width=\linewidth]{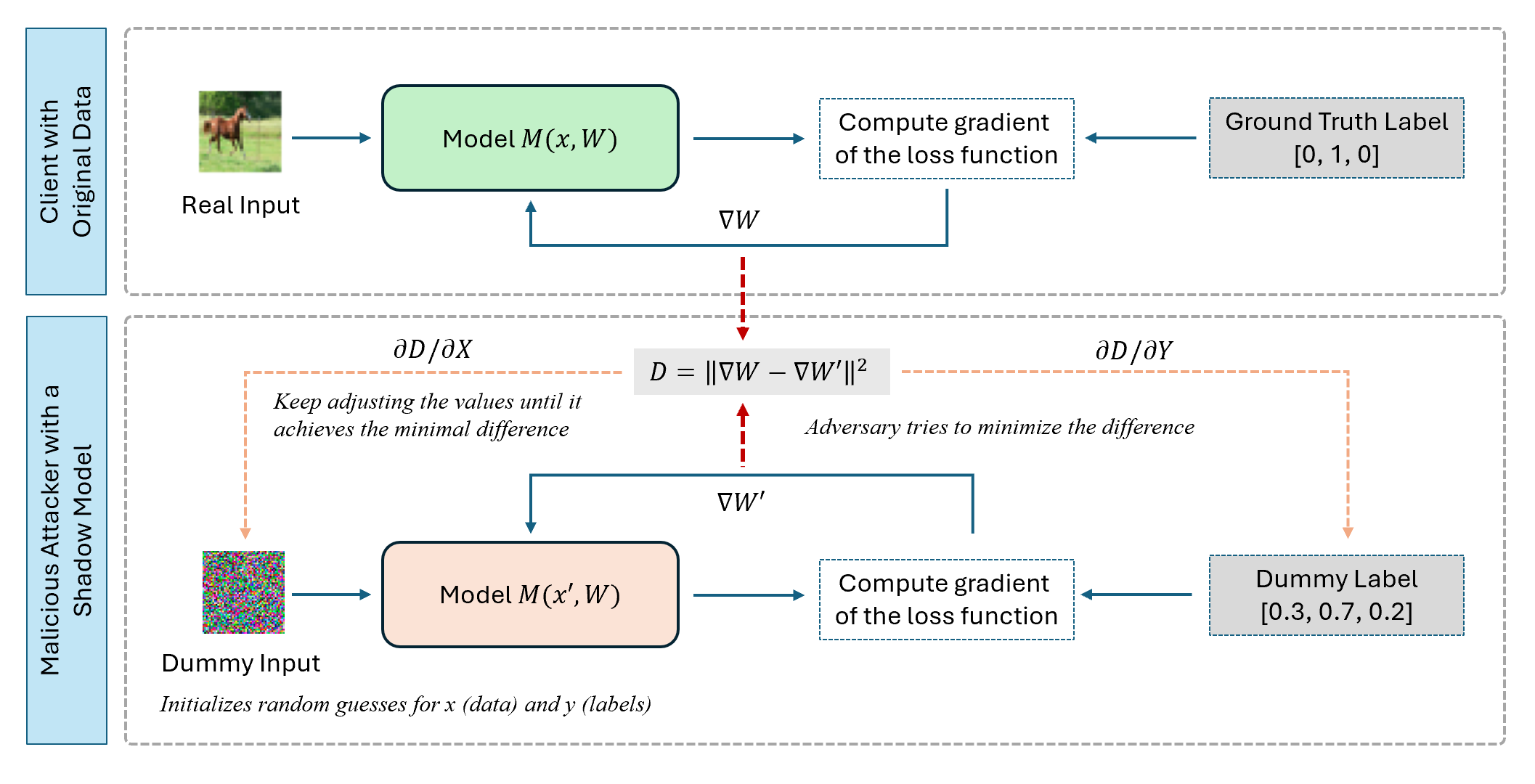} % adjust path & size
    \caption{llustration of a DLG attack, where an adversary reconstructs sensitive client data from shared gradients.}
    \label{fig:DLG_attack}
\end{figure}

Among these threats, the model inversion attack, also known as Deep Leakage from Gradients (DLG)~\cite{b1}, poses a significant risk to federated learning systems. In this attack, a malicious server aims to reconstruct client data by exploiting gradient information from the aggregated or local model updates. As shown in Fig. \ref{fig:DLG_attack}, the server trains a secondary model that iteratively adjusts dummy data to match the extracted gradients using gradient descent. By minimizing the difference between the dummy gradients and the target gradients, the server can approximate the original client data, effectively breaching the privacy protections of FL. Recent research has further refined these attack techniques, demonstrating their evolving sophistication~\cite{b2, b3}. In the original DLG attack, the authors proposed using Euclidean distance as the loss function to measure the similarity between reconstructed and target gradients, proving the feasibility of reconstructing training data in simpler network architectures. Building upon this, Geiping et al.~\cite{b3} introduced an improved metric based on cosine similarity, which considers both the magnitude and direction of gradient changes. This refinement enabled more effective reconstruction of images in deeper neural network architectures, highlighting the growing capabilities of such attacks.

In summary, these studies highlight significant privacy vulnerabilities in standard federated learning systems, emphasizing the urgent need for stronger privacy-preserving mechanisms. The demonstrated ability to reconstruct individual training samples from aggregated updates challenges the core privacy assumptions of FL, underscoring the importance of developing more robust defenses against such attacks.

\begin{figure}[H]
    \centering
    \includegraphics[width=\linewidth]{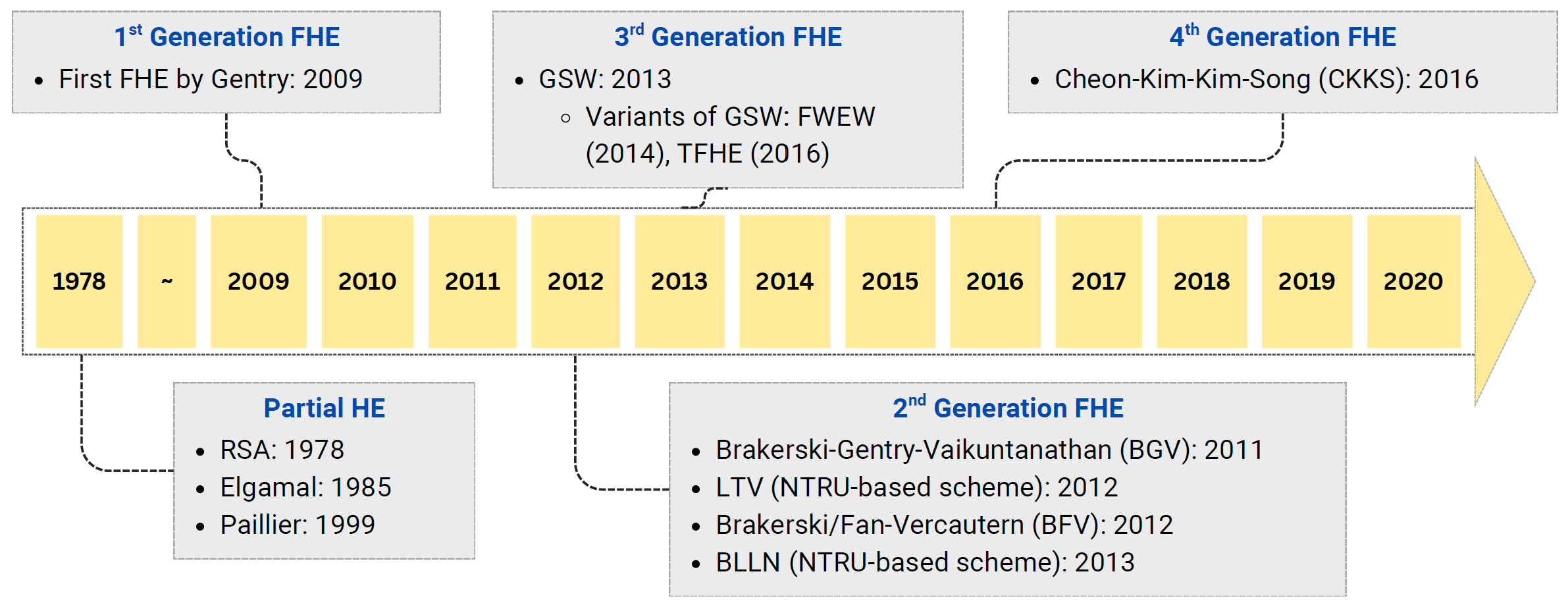} % adjust path & size
    \caption{Four generations in the evolution of homomorphic encryption schemes.}
    \label{fig:HE_schems}
\end{figure}

\subsection{Enhancing Privacy and Security with Homomorphic Encryption}
\label{Enhancing Privacy and Security in FL Systems through Homomorphic Encryption}

Despite the associated computational and communication overhead, homomorphic encryption has emerged as a promising solution for mitigating privacy and security threats in FL systems. Although not originally designed with deep learning algorithms in mind, HE enables computations to be performed directly on encrypted data (ciphertexts), eliminating the need for decryption during intermediate computational stages. This capability allows seamless integration into FL environments, where clients can securely transmit encrypted local model updates to a central server. The server, in turn, aggregates these ciphertexts without ever accessing the underlying plaintext data, thereby significantly reducing the risk of sensitive information leakage to potentially malicious servers or adversaries.

Formally, an HE scheme consists of four fundamental algorithms:

\begin{itemize}
    \item \(\textit{HE.KeyGen}(\lambda) \rightarrow (\textit{pk}, \textit{sk})\): Generates a public key \(\textit{pk}\) and a secret key \(\textit{sk}\), given a security parameter \(\lambda\).
    \item \(\textit{HE.Enc}(\textit{pk}, m) \rightarrow c\): Encrypts a plaintext message \(m\) into ciphertext \(c\) using the public key \(\textit{pk}\).
    \item \(\textit{HE.Eval}(f, c_1, c_2, \dots, c_n) \rightarrow c'\): Computes a function \(f\) directly on ciphertexts \(\{c_i\}\), producing ciphertext \(c'\), equivalent to encrypting \(f\) applied to the plaintexts.
    \item \(\textit{HE.Dec}(\textit{sk}, c') \rightarrow m\): Decrypts ciphertext \(c'\) back into the original plaintext \(m\) using the secret key \(\textit{sk}\).
\end{itemize}

Over the past few decades, homomorphic encryption schemes have advanced significantly, with many now designed to offer protection even against quantum threats. The following sub section provides a brief categorization of these schemes based on their supported arithmetic operations and their suitability for practical deployments.

\subsubsection{Partially Homomorphic Encryption (PHE)}
PHE schemes generally permit unlimited execution of only a single arithmetic operation (either addition or multiplication) on encrypted data. For instance, the Paillier cryptosystem supports unlimited additive operations, making it suitable for cumulative computations common in privacy-preserving machine learning~\cite{b14,b15,b16}. However, the Paillier cryptosystem's security relies on the Decisional Composite Residuosity (DCR) problem, a factoring-based assumption~\cite{b17}. Due to Shor's algorithm, Paillier is vulnerable to quantum attack~\cite{b18}.  

ElGamal is another PHE scheme that inherently supports multiplicative operations~\cite{b19}. To adapt ElGamal for privacy-preserving federated learning with additive operations, several variants have been proposed. In their work, Chen et al. presented~\cite{b20} a lightweight approach utilizing multi-key EC-ElGamal cryptosystem, which generates shorter keys and ciphertexts, making it particularly suitable for resource-constrained IoT environments. However, since ElGamal's security relies on the hardness of the discrete logarithm problem, it remains vulnerable to quantum attacks, as demonstrated by Shor's algorithm~\cite{b21, b22}.

\subsubsection{Fully Homomorphic Encryption (FHE)}
In contrast, FHE schemes support an unlimited number of additions and multiplications on ciphertexts, greatly extending their applicability to complex computational tasks like machine learning. The prominent FHE schemes include BFV~\cite{b23, b24}, BGV~\cite{b25}, TFHE~\cite{b26}, and CKKS~\cite{b27}. Each scheme offers distinct advantages tailored to specific computational contexts:

\begin{itemize}
    \item \textbf{BFV and BGV} are optimized for integer arithmetic, particularly suited for integer-based computations. However, machine learning commonly involves floating-point numbers, requiring additional preprocessing steps like quantization, rescaling, or rounding to adapt model parameters into integers compatible with these schemes. For example, previous research~\cite{b28} employs quantization, and~\cite{b29} uses rescaling and rounding methods to address this compatibility issue.

    \item \textbf{TFHE} specializes in efficient Boolean and bitwise computations. When applying TFHE to privacy-preserving machine learning, model parameters must be converted into binary representations. Recent studies, such as~\cite{b30}, address this by quantizing model weights to integers, subsequently encoding these integers into binary formats suitable for TFHE operations.

    \item \textbf{CKKS} uniquely handles real and complex numbers, making it particularly advantageous for machine learning models utilizing floating-point arithmetic, such as deep neural networks. CKKS achieves performance gains through approximate computations, tolerating minor numerical inaccuracies for significantly improved computational efficiency. Moreover, CKKS supports encoding multiple values into a single ciphertext, facilitating effective batch processing—essential for federated learning scenarios.
\end{itemize}

Modern FHE schemes (BFV, BGV, TFHE, and CKKS) are based on the Ring Learning With Errors (RLWE) problem~\cite{b23, b24, b25, b26, b27}, a mathematical assumption associated with high-dimensional lattices. RLWE is currently considered quantum resilient, providing stronger security guarantees against quantum computing threats compared to factoring-based schemes like Paillier. Recent literature, including Kara et al.~\cite{b31}, reinforces RLWE's quantum resilience, underscoring its suitability for robust and future-proof cryptographic applications in privacy-sensitive domains such as federated learning.

In this work, we address model inversion attacks, specifically DLG attack. We leverage a CKKS-based homomorphic encryption scheme with a selective parameter encryption technique to enhance privacy protection while mitigating the computational overhead introduced by homomorphic encryption. Our contribution includes developing a theoretical framework to analyze the effectiveness of selective encryption defenses and identify key factors that influence their performance. To the best of our knowledge, this is the first work to provide a theoretical analysis of the effectiveness of selective encryption in defending against model inversion attacks.

%=======================================================================
\section{Background and Related Work} \label{RelatedWork}
This section starts with reviewing HE–based techniques for protecting client privacy in federated learning, with emphasis on defenses against gradient inversion and data reconstruction attacks. We then survey the theoretical results on reconstruction quality in these attacks. We conclude by detailing the threat model that underpins our experiments and analysis. %We also discuss theoretical foundations for evaluating privacy guarantees, specifically using the Bayesian Cramér-Rao lower bound.

\subsection{Defending Reconstruction Attacks with HE-Based Selective Encryption} \label{HE-based Defense Mechanisms}
While Homomorphic Encryption has shown effectiveness in protecting data privacy, its considerable computational overhead limits its practical deployment in real-world scenarios. To address this, recent studies propose selective encryption strategies that encrypt only sensitive or critical model parameters using the CKKS encryption scheme~\cite{b4,b5,b32}. These selective approaches aim to strike a balance between security and computational efficiency by identifying and encrypting the most privacy-sensitive model parameters. For instance, some of the existing approaches identify sensitive model parameters based on second-order gradient derivatives or the magnitude of loss gradients. Specifically, FedML-HE~\cite{b4} employs second-order derivatives by assessing how the loss gradient with respect to model parameters varies in relation to the ground truth labels. Another recent method, referred to as Magnitude-based encryption~\cite{b5}, selects critical parameters based purely on the magnitude of the gradients, encrypting parameters with higher sensitivity scores.

In a related context but with a different approach, the MaskCrypt~\cite{b32} introduces a gradient-guided masking algorithm, allowing each client to dynamically determine which parameters are crucial and should be encrypted during training. Although selective encryption methods have demonstrated improved efficiency, a major limitation persists: existing studies are almost entirely empirical, with little to no rigorous theoretical analysis to substantiate their privacy guarantees and effectiveness.

\begin{figure*}[ht]
\centering
\centerline{\includegraphics[scale=0.4]{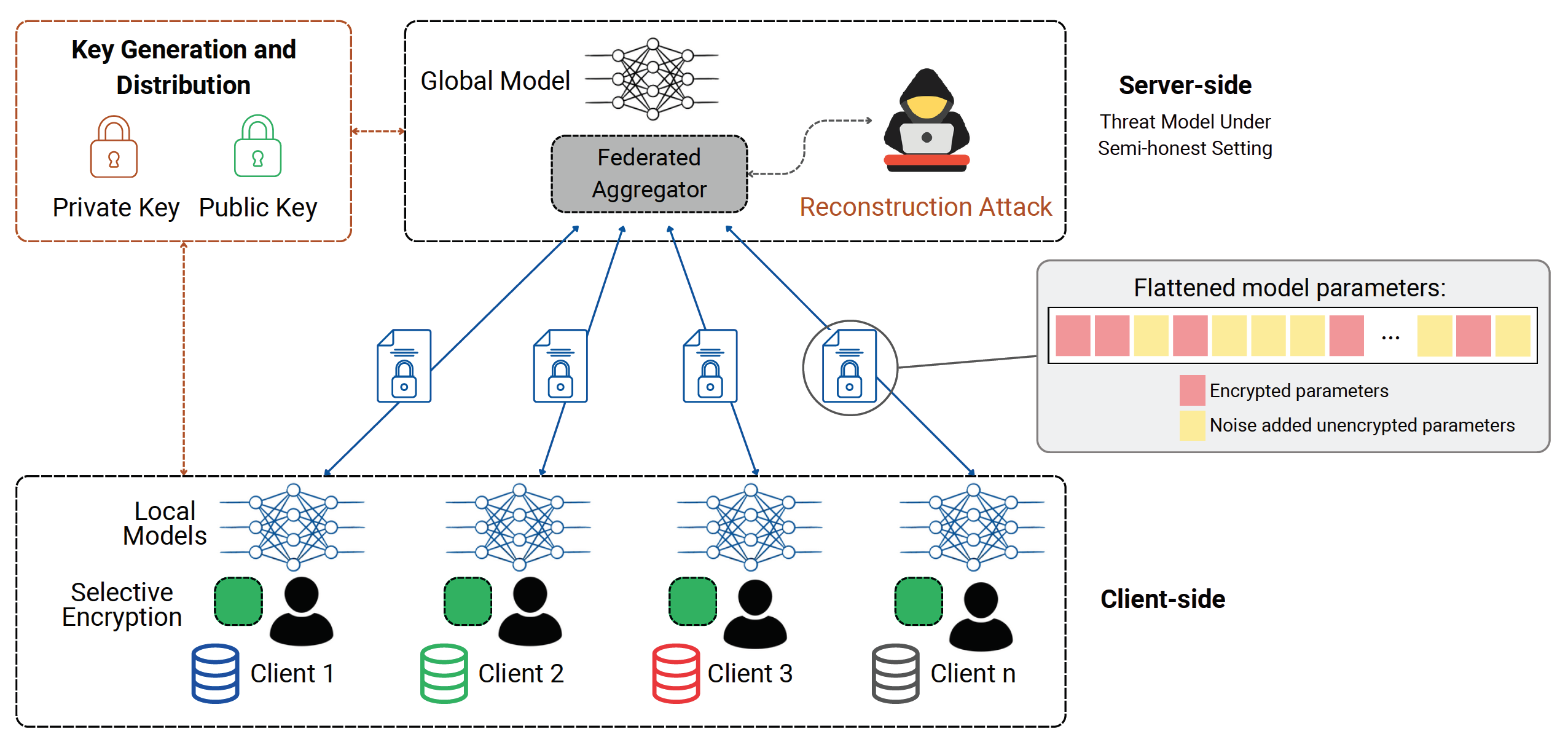}}
\caption{System workflow of the HE-based selective encryption framework. The attacker only gains access to unencrypted parameters, while the framework aims to balance security and utility. }
\vspace*{-5pt}
\label{fig:flowchart}
\end{figure*}

\subsection{Bounding Reconstruction Error in Gradient Leakage Attacks}
Model inversion attacks attempt to reverse-engineer training data by exploiting access to model updates or gradients. To theoretically assess the robustness of privacy-preserving methods against such attacks, different approaches have been introduced under a variety of assumptions to analyze the data reconstruction error. \cite{b36} showed that for a 
two-layer network with $m$ hidden nodes and input dimension $d$, if gradients are 
perturbed with Gaussian noise $\epsilon \sim \mathcal{N}(0,\sigma^{2}I_{md+m})$, 
then with high probability the reconstruction risk satisfies
$
R_{L} \;\;\geq\;\; \widetilde{\Omega}\!\left(\tfrac{\sigma d}{m}\right).
$ This establishes that higher 
noise level $\sigma$ and higher input dimension $d$ makes recovery harder, while 
wider networks ($m$ large) make recovery easier. Looking from an approximation perspective, this result provided a fundamental 
information-theoretic lower bound on data identifiability from noisy gradients. The paper also further analyzed the effect of pruning-based 
defenses (which includes defenses like homomorphic encryption that make true gradient values inaccessible). Again assuming the same two-layer network coupled with a gradient pruning defense with pruning ratio $p$, they showed that with high probability, the 
statistical reconstruction risk is lower bounded by
$
\widetilde{\Omega}\!\left(\tfrac{\sigma d}{(1-p)m}\right)
$. This 
result highlighted that pruning reduces the number of gradient coordinates 
available to an attacker from $m$ to $(1-p)m$, thereby amplifying the 
privacy guarantee by raising the fundamental lower bound on reconstruction 
error. Continuing with the same line of thinking, Chen et al.~\cite{b34} showed that the reconstruction error is
lower bounded by a quantity proportional to
$d^{2}\!/\!\big[\mathbb{E}_{x\sim \mathcal{X}}\!\operatorname{tr}(J_{F}(x)) + d\,\lambda_{1}(J_{P})\big]$.
Here $d$ is the input dimension; $J_{F}(x) := \mathbb{E}_{y\sim S(g(x))}\!\big(\nabla_{x}\log p_{S(x)}(y)\,
\nabla_{x}\log p_{S(x)}(y)^{\top}\big)$ is a Fisher-information–like matrix for the defended
gradient channel $S$ applied to the model output $g(x)$ (with $p_{S(x)}(y)$ the defense-induced
observation density); $J_{P} := \mathbb{E}_{x\sim \mathcal{X}}\!\big(\nabla_{x}\log p_{\mathcal{X}}(x)\,\nabla_{x}\log p_{\mathcal{X}}(x)^{\top}\big)$
captures population-level sensitivity of $\mathcal{X}$ (with $p_{\mathcal{X}}$ the data density); $\operatorname{tr}(\cdot)$
is the trace operator and $\lambda_{1}(\cdot)$ the largest eigenvalue. Intuitively, larger $d$ increases
the bound (stronger privacy), whereas more informative defended gradients (large $\operatorname{tr}(J_{F})$)
or stronger population structure (large $\lambda_{1}(J_{P})$) shrink it, indicating weaker privacy.
This links defense design (noise, clipping shaping $p_{S(x)}$) to optimization-time leakage dynamics through explicit information terms. Developing on this result from ~\cite{b34}, in the current research we quantify the relation of reconstruction error with model complexity and encryption ratio, thereby providing insight into model specific risk behaviors.

\subsection{Threat Model for Theoretical Analysis} 
We consider a semi-honest (honest-but-curious) server that follows the federated learning protocol for aggregation and distribution but actively attempts to infer client's private information from model updates. The server can store all client updates separately and analyze them individually or collectively across training rounds. We assume that the server cannot modify the model architecture to facilitate attacks, ensuring that the selective encryption mechanism remains effective against the model inversion attack without altering the core FL framework.

% TODO (Before TIFS): introduce dynamic noise

%=============================================================================
\section{Preliminaries and Workflow of the Proposed Framework} \label{preliminaries}
This section defines essential operators and outlines the workflow underlying our theoretical analysis concerning the lower bound on data reconstruction error.

\subsection{Notations and Definitions}

For the considered gradient leakage attack, we define the notation explicitly as follows:
Note: $m$: data point dimension; $D$: dimension of model (gradients); $d$: dimension of unencrypted gradients
\begin{itemize}
    \item \(x \sim \mathcal{X} \in \mathbb{R}^m\): Original data sample.
    \item $\theta \in \mathbb{R}^D$: Model parameters.
    \item $L(x) \in \mathbb{R}$: Loss function at sample $x$.
    \item \(g(x) = \nabla_\theta L(x) \in \mathbb{R}^D\): Vector of gradients with respect to model parameters.
    \item \(Rg(x) \in \mathbb{R}^d\): Reduced gradient vector obtained by removing encrypted gradient entries. Note that \( d \leq D \). Here $R:\mathbb{R}^D \ \to \mathbb{R}^d$ is a restriction operator.
    \item \(u = Rg(x) + \delta \in \mathbb{R}^d\): Unencrypted gradient entries with additive noise:
    \[
        \delta \sim \mathcal{N}(0, \sigma^2 I).
    \]
    \item \(y = Pu \in \mathbb{R}^D\): The final gradient vector observed by the attacker consists of noisy unencrypted gradients and zeros at encrypted positions, effectively reconstructing a $D$-dimensional vector where encrypted components are replaced with zeros.Here $P: \mathbb{R}^d \to \mathbb{R}^D$ is a prolongation operator.
    \item \(A(y) \in \mathbb{R}^m\): Reconstruction algorithm attempting to recover the original data \( x \) from observed noisy and partially-encrypted gradient vector \( y \).
\end{itemize}

\begin{comment}
\subsection{Operators}

Two key linear operators underpin our workflow:
\begin{itemize}
    \item \textbf{Restriction operator \(
        R : \mathbb{R}^D \rightarrow \mathbb{R}^d
    \)}:    
    removes encrypted gradient locations from \(g(x)\).

    \item \textbf{Prolongation operator \(
        P : \mathbb{R}^d \rightarrow \mathbb{R}^D
    \)}:
    extends the reduced gradient vector by adding zeros at encrypted gradient locations, restoring its original dimension.
\end{itemize}
\end{comment}

It should be noted that both operators \(R\) and \(P\) can be explicitly represented as binary matrices composed solely of \(0\)'s and \(1\)'s.

\subsection{Workflow of the Selective Encryption Mechanism}
To defend against model inversion attacks that attempt to reconstruct data from gradients, we propose a workflow that combines selective gradient encryption with Gaussian noise. With the established notations, our entire workflow can be represented as:
\begin{equation*}
    \small
    x \rightarrow g(x) \rightarrow Rg(x) \rightarrow u = Rg(x) + \delta \rightarrow y = Pu \rightarrow A(y).
\label{eq:workflow}
\end{equation*}

\begin{enumerate}
    \item The gradient $g(x) \in \mathbb{R}^D$ is computed from the model using training data $x$.
    \item The restriction operator $R$ selects only the unencrypted components, removing encrypted gradients from their corresponding positions in $g(x)$.
    \item Gaussian noise $\delta$ is added to the unencrypted gradients, yielding the noisy vector $u = Rg(x) + \delta$.
    \item The prolongation operator $P$ reconstructs a $D$-dimensional vector by inserting zeros at encrypted positions, restoring the original dimensionality.
\end{enumerate}

The final gradient vector \( y \) observed by the attacker consists of a combination of encrypted gradients and noisy unencrypted gradients, thereby complicating the reconstruction task.

Since the operator $P$ merely inserts zeros without altering the unencrypted gradient information, reconstruction from $y = Pu$ yields identical results to reconstruction from $u$. Therefore, for denoting data reconstruction we use the notation $A(u)$ and $A(y)$ interchangeably in the remainder of the text.

% \begin{equation*}
% A(y) = A(u) \text{ if } y = Pu
% \end{equation*}

\subsection{Observed distribution under Selective encryption}
Given \( u = Rg(x) + \delta \), where \(\delta \sim \mathcal{N}(0, \sigma^2 I)\), it follows that:
\[
    u \sim \mathcal{N}(Rg(x), \sigma^2 I).
\]

Since \( y = Pu \), we derive the distribution explicitly as:
\begin{equation*}
    y \sim \mathcal{N}(PRg(x), \sigma^2 PP^T).
    \label{eq:y_distribution}
\end{equation*}

From this point onward, we denote this distribution by:
\[
    y \sim S(g(x)),
\]
explicitly representing the selective encryption-based defense mechanism \( S(\cdot) \). Thus, the operator \(S(\cdot)\) encapsulates three distinct actions: 1) it encrypts selected gradients, 2) adds Gaussian noise to the unencrypted gradients, and 3) finally projects the reduced gradients back to the original dimension \(\mathbb{R}^{D}\).

\subsection{Reconstruction Error} \label{theoretical_reconstruction_lower_bound}
Formally, the reconstruction error \(E_A\) can be defined as the minimum expected squared reconstruction error:
\begin{equation}
    E_A = \min_{A:\mathbb{R}^{D}\rightarrow\mathbb{R}^{m}}
    \mathbb{E}_{x\sim X}\mathbb{E}_{y\sim S(g(x))}
    [\|A(y)-x\|_2^2].
    \label{eq:reconstruction_error}
\end{equation}
With that, we now formalize our usage of Bayesian Cramér-Rao Lower Bound (BCRLB)~\cite{b33} in lower bounding $E_A$.

Given a noisy observation $u$, BCRLB provides a lower bound for data ($x$) reconstruction error for any estimator \( A(\cdot) \): 
\begin{equation}
    \mathbb{E}_{x,u}[(A(u)-x)(A(u)-x)^T] \geq J_B^{-1},
    \label{eq:bcrlb_main}
\end{equation}
where \(J_B\) is the Bayesian Fisher information explicitly defined as:
\begin{equation}
    J_B = J_P + J_D,
    \label{Jb}
\end{equation}

with:
% \begin{align*}
%     J_P &= \mathbb{E}_{x,u}\left[\nabla_x \log p(x,u) \nabla_x \log p(x,u)^T\right], \\
%     J_D &= \mathbb{E}_{x \sim D}\left[J_F(x)\right], \\
%     J_F(x) &= \mathbb{E}_{u|x}\left[\nabla_x \log p(u|x) \nabla_x \log p(u|x)^T\right]
% \end{align*}

% \begin{equation}
%     J_P = \mathbb{E}_{x,u}\left[\nabla_x \log p(x,u) \nabla_x \log p(x,u)^T\right]
%     \label{eq:bcrlb_JP}
% \end{equation}

\begin{equation}
    J_P = \mathbb{E}_{x \sim \mathcal{X}}\left[\nabla_x \log p_{\mathcal{X}}(x) \nabla_x \log p_{\mathcal{X}
    }(x)^T\right],
    \label{eq:bcrlb_JP}
\end{equation}

\begin{equation}
    J_D = \mathbb{E}_{x \sim D}\left[J_F(x)\right]
    \label{eq:bcrlb_JD}
\end{equation}
where, 
\begin{equation}
    J_F(x) = \mathbb{E}_{u|x}\left[\nabla_x \log p(u|x) \nabla_x \log p(u|x)^T\right].
    \label{eq:bcrlb_JF}
\end{equation}
Here $J_P$ and $J_D$ are prior informed and data informed terms respectively. With $x$ and $y$ as defined before , and assuming an estimator $A(y)$ designed to recover $x$ from the observations, applying the Bayesian Cramér-Rao bound requires satisfying the following regularity conditions \cite{b33}:

\begin{itemize}
    \item \textbf{Support Condition:}  
    The distribution \( \mathcal{X} \) spans either the entirety of \( \mathbb{R}^m \) or an open region within \( \mathbb{R}^m \) with a piecewise smooth boundary.

    \item \textbf{Existence of Derivatives:}  
    All partial derivatives of the joint probability density \( p(x, y) \) with respect to \( x \) exist and are absolutely integrable over the respective domains.

    \item \textbf{Finite Bias Condition:}  
    The estimator bias defined as:
    \[
    B(x) = \int (A(y) - x) p(y | x) \, dy
    \]
    remains finite for every \( x \in \mathbb{R}^m \).

    \item \textbf{Interchangeability of Derivative and Integral:}  
    For all \( x \in \mathbb{R}^m \), it must hold that:
    \[\scriptstyle
        \nabla_x \int p(x, y) [A(y) - x]^T \, dy 
        = 
        \int \nabla_x \left[ p(x, y) [A(y) - x]^T \right] dy
    \]
    
    \item \textbf{Boundary Condition for Estimation Error:}  
    As points \( x_i \) within the support \( \text{supp}(\mathcal{X}) \) approach boundary points \( x \), the product \( B(x_i) p(x_i) \) must approach zero.
\end{itemize}
Assuming these conditions are satisfied for the data spaces and the probability measures over them, we now present our main theoretical results.

%===================================================================
\section{Analysis of Gradient Distribution and Reconstruction Error}  \label{intro_BCRLB}
This section builds on the statistical properties of the observed noisy-encrypted gradient vector $y$ and develops a lower bound on the corresponding data reconstruction error. The BCRLB utilizes Fisher Information to quantify the amount of information that an observed random variable carries about an unknown parameter. To apply the BCRLB for establishing a lower bound on model inversion attack performance when clients employ selective encryption, the attacker's observed gradients must be treated as random variables rather than deterministic values.

This probabilistic framework is necessary because Fisher Information is defined as the expected value of the squared score function, which requires a probability distribution over the observed data. Specifically, Fisher Information is computed as $I(\theta) = \mathbb{E}\left[\left(\frac{\partial \log p(y|\theta)}{\partial \theta}\right)^2\right]$, where $p(y|\theta)$ represents the likelihood of observing data $y$ given parameter $\theta$. When gradients are deterministic (fixed values), there is no underlying probability distribution, making the likelihood function undefined and rendering Fisher Information mathematically meaningless.

Therefore, we model the observed gradients as random variables by introducing small random noise from a known distribution to the unencrypted gradients. This transforms them from fixed unknown values into probabilistic quantities with well-defined likelihood functions, enabling the computation of Fisher Information and subsequent application of BCRLB for lower bound analysis. 

Moving beyond mathematical utility, noise addition also has a practical significance form privacy protection point of view. In a nutshell, the additive noise enhances privacy by obfuscating gradient information and limiting attackers' reconstruction capabilities. Even with knowledge of model architecture and gradients, perfect training data recovery becomes infeasible. This approach complements selective encryption, creating a multi-layered defense: encrypted positions provide complete protection while noisy unencrypted positions preserve utility with probabilistic privacy guarantees against model inversion attacks.

%================================================================================
\subsection{Reconstruction Error Lower Bound Analysis Using the Fisher Information Matrix} \label{derive_bound}
Building on the work of Chen et al.~\cite{b34}, here we further derive the specific reconstruction error lower bound for the attacker's observed gradients under selective encryption. We begin by formulating the Fisher information matrix $J_F(x)$ as shown in Equation \ref{eq:bcrlb_JF}. Given the normally distributed gradient vector $u \sim \mathcal{N}(Rg(x), \sigma^2 I)$, we obtain:
\begin{equation*}
J_F(x) = \frac{1}{\sigma^2} \nabla_x g(x)^T R^T R\nabla_x g(x) \in \mathbb{R}^{m \times m} 
\end{equation*}
The proof is included as Lemma 2 in appendix. 

Using this Fisher information matrix, we can now establish the reconstruction error lower bound:

\begin{theorem} \label{theo:theorem1}
(Privacy guarantee) With $E_A$ being the mean data reconstruction error in a gradient inversion attack as defined in (\ref{eq:reconstruction_error}), we have the following respective lower bound:
\begin{equation} \label{Equation_6}
    \small
        E_A \geq 
        \frac{m}{\frac{D(1 - z)}{\sigma^2} \mathbb{E}_{x \sim \mathcal{X}} \| R \nabla_x g(x) \|^2_{\max} + \lambda_1(J_P)},
\end{equation}
where $m$ is the data dimension, $D$ is the number of model parameters, $z$ is the model encryption ratio, $\sigma^2$ is the variance for the additive noise defense, $\lambda_1(J_P)$ is the largest eigenvalue of the prior Fisher information matrix shown in (\ref{eq:bcrlb_JP}), and $\| R \nabla_x g(x) \|^2_{\max} = (\max_{i,j}{|R \nabla_xg(x)|})^2$ quantifies the maximum gradient exposure defined as squared maximum sensitivity of an unencrypted gradient coordinate $g(x)_j$ with respect to a data feature $x_i$.
\end{theorem}

% We begin by formulating the Fisher information matrix $J_F(x)$ as shown in Equation \ref{eq:bcrlb_JF}. 
% % Since reconstruction from $y$
% %  and $u$ yields equivalent results (both contain identical gradient information), we simplify the analysis by working with $A(u)$ where:

% % \begin{equation*}
% % A(y) = A(u) \text{ if } y = Pu
% % \end{equation*}

% For the normally distributed gradient vector $u \sim \mathcal{N}(Rg(x), \sigma^2 I)$, we derive:

% \begin{equation*}
% J_F(x) = \frac{1}{\sigma^2} \nabla_x g(x)^T R^T R\nabla_x g(x) \in \mathbb{R}^{m \times m}
% \end{equation*}

% % where $E_A$ represents the reconstruction error:

% % \begin{equation}
% % E_A = \min_{A:\mathbb{R}^D \rightarrow \mathbb{R}^m} \mathbb{E}_{x \sim X}\mathbb{E}_{y \sim S(g(x))}\|A(y) - x\|^2
% % \end{equation}

Detailed derivations for Theorem~\ref{theo:theorem1} are provided explicitly in Appendix ~\ref{derive_reconstruction_lower_bound}. For a typical FL setup, Theorem \ref{theo:theorem1} has a variety of implications regarding the underlying AI model complexity, encryption hyperparameters etc., that we discuss in the following subsection.

\subsection{Analysis of Factors Influencing the Reconstruction Error Bound}
The derived reconstruction error bound (\ref{Equation_6}) is influenced by several key factors that shape the effectiveness of our defensive encryption strategy. Specifically, it is governed by the encryption ratio $z$, which represents the proportion of encrypted parameters, the total number of model parameters $D$, and the expected maximum gradient exposure $\mathbb{E}_{x \sim X} \left[ \left\| R \nabla_x g(x) \right\|_{\max}^2 \right]$. This section examines how these parameters affect the reconstruction error bound and discusses their broader security implications.

\subsubsection{Impact of Encryption Ratio}
When the encryption ratio $z$ increases, the denominator term $D(1-z)$ decreases, resulting in a higher reconstruction error lower bound:

\begin{equation*}
z \uparrow \implies D(1-z) \downarrow \implies E_A\text{ lower bound} \uparrow
\end{equation*}

This mathematical relationship confirms that increasing the proportion of encrypted parameters strengthens privacy protection by making data reconstruction more difficult for potential attackers.

\subsubsection{Effect of Model Complexity}
For a fixed encryption ratio $z$, reducing the total parameter count $D$ leads to an increased reconstruction error bound:

\begin{equation*}
D \downarrow \implies E_A\text{ lower bound} \uparrow
\end{equation*}

This indicates that smaller models with fewer parameters lead to a higher reconstruction error bound, making it harder for the attacker to reconstruct the data. This is because of decreased number of gradient coordinates to extract information from.

\subsubsection{Gradient exposure Considerations}
The expected maximum gradient exposure $\mathbb{E}_{x \sim X}[\|R\nabla_x g(x)\|^2_{\max}]$ represents how strongly unencrypted gradients are influenced by input data. A lower sensitivity leads to a higher reconstruction error bound, thereby strengthening privacy guarantees.

\begin{equation*}
\mathbb{E}_{x \sim X}[\|R\nabla_x g(x)\|^2_{\max}] \downarrow \implies E_A\text{ lower bound} \uparrow
\end{equation*}

In other words, this relationship shows that by selectively encrypting the most sensitive gradient components (i.e., those with the largest magnitudes) and leaving the less sensitive ones unencrypted, the reconstruction error bound can be substantially increased. 

%================================================
\section{Effect of Selective HE Combined with Noise on Model Utility} \label{utility_reduction}

Since our proposed approach and analysis introduces small amounts of noise into the unencrypted gradients, in this section, we examine its effect on the training loss (and the final utility) over a single gradient descent step in standard FL training. This analsysis essentially points to the classic privacy-utility trade-off that becomes relevant when trying to encrypt or polute shared gradients for increased privacy.

%\subsection{Evaluating Loss Reduction Across One Epoch With and Without Selective Encryption}
\subsection{Performance Implications of Selective Noise Addition to Unencrypted Gradients}
Although higher variance ($\sigma^2$) of the additive noise increases the lower bound on reconstruction error as established by Theorem \ref{theo:theorem1}, thereby directly improving client privacy, too much noise in the locally computed gradient can hamper the training process till the point of losing utlity. To understand this trade-off better, in this section we analyze the reduction in local loss $L_i$ of a general client $i$ with one gradient descent step. The motivation behind this analysis is to establish (approximate) conditions under which the local model can diverge. In this regard we have the following result on the noise variance.
\begin{theorem}
\label{Theorem 2}
(Privacy-utility trade-off) Under the first order Taylor approximation of loss $L_i(x, \theta) \in \mathbb{R}$ for client $i$ at $\theta$ (model parameters) along the local gradient descent step $\theta^+=\theta-\eta Q(x)$, the expected loss reduces ($\mathbb{E}_x\!\big[L_i(x,\theta)-L_i(x,\theta^{+})\big] \geq 0$) with probability at least $1-\delta$ $($for any $\delta\in(0,1))$, if for all clients, the variance ($\sigma^2$) for the noise added before last aggregation was upper bounded as:
\begin{equation} \label{Equation_critical_noise}
\sigma \;\leq \; \sigma_{\mathrm{crit}}
:= \frac{B_i}{\sqrt{n}\,\|\mu_i\|\,\big(\sqrt{d}+\sqrt{2\ln(1/\delta)}\big)}.
\end{equation} 
Here
\begin{itemize}
    \item $n$, $d$: Number of clients and number of noisy gradient coordinates respectively.
    \item $Q(x)=G(x)+P\!\sum_{i=1}^n \epsilon_i$ is the aggregated, partially noisy gradient returned from the server to the clients, with $\epsilon_i \in \sim N(0,\sigma^2 I_d)$, $I_d \in \mathbb{R}^{d \times d} $.
    \item $G(x) = \sum_{k=1}^n g_k(x)$ represents the clean aggregated gradients from clients. For client $i$: $g_i(x)=\nabla_{\theta} L_i(x,\theta)$. 
    \item $R \in \mathbb{R}^{d \times D}$ is the restriction operator sampling $d$ unencrypted gradients from a given gradient vector, and $P$ being the corresponding prolongation operator ($P = R^T$).
    \item $\mu_i:=\mathbb{E}_x[\,R\,g_i(x)\,]\in\mathbb{R}^d$ is the restriction of local gradient to coordinates where noise is added.
    \item $B_i:=\mathbb{E}_x[\,g_i(x)^\top G(x)\,] \in \mathbb{R}$ is the correlation measure of the local gradient $g_i(x)$ with the global aggregated gradient $G(x)$. Please note if $B\le 0$, ascensive steps may occur even for $\sigma=0$.
\end{itemize}
\end{theorem}
Derivation for Theorem \ref{Theorem 2} is provided in Appendix \ref{derive_privacy_utility}.
\subsection{Interpreting the Noise Threshold in Federated Learning}
The theorem provides a high-probability, one-step descent guarantee for each client’s local loss when the server returns an aggregated gradient that is perturbed only on a subset of coordinates. It identifies a critical noise level
\[
\sigma_{\mathrm{crit}}
=\frac{B_i}{\sqrt{n}\,\|\mu_i\|\,\big(\sqrt{d}+\sqrt{2\ln(1/\delta)}\big)}
\]
such that, with probability at least $1-\delta$ over the injected Gaussian noise (additive Gaussian noise), the first-order predicted local loss of client $i$ decreases whenever $\sigma\le\sigma_{\mathrm{crit}}$.
\begin{enumerate}
\item Privacy--utility trade-off
The bound makes explicit how utility competes with privacy:
\begin{itemize}
    \item \textbf{Signal (alignment)}: $B_i=\mathbb{E}_x\!\big[g_i(x)^\top G(x)\big]$ measures how well the client’s local gradient aligns with the global aggregate. Larger $B_i$ increases the tolerable noise.
    \item \textbf{Where noise is injected}: $\mu_i=\mathbb{E}_x\!\big[R\,g_i(x)\big]$ captures the average magnitude of the local gradient on the \emph{noised} coordinates. The threshold shrinks as $1/\|\mu_i\|$: if noise targets coordinates carrying strong local signal, less noise can be tolerated.
    \item \textbf{Number of clients}: the server \emph{sums} independent noises, so the aggregated noise grows like $\sigma\sqrt{n}$; hence $\sigma_{\mathrm{crit}}$ scales as $1/\sqrt{n}$.
    \item \textbf{Number of noisy coordinates}: the norm of a $d$-dimensional Gaussian scales like $\sqrt{d}$, yielding the $\sqrt{d}$ factor in the denominator and tightening the threshold as more coordinates are noised.
    \item \textbf{Confidence level}: smaller $\delta$ (stronger guarantee) increases $\sqrt{2\ln(1/\delta)}$ and thus reduces $\sigma_{\mathrm{crit}}$.
\end{itemize}
\item Geometric intuition
We have the first order loss reduction estimate,
\begin{align*}
\mathbb{E}_x\!\big[L_i(x,\theta)-L_i(x,\theta^{+})\big]
&\approx
\eta\,\underbrace{\mathbb{E}_x\!\big[g_i(x)^\top G(x)\big]}_{\text{signal }B_i}
\\[2pt]
&\quad+\;\eta\,\underbrace{\mathbb{E}_x\!\big[g_i(x)\big]^{\!\top}\,P\!\sum_{j=1}^n \epsilon_j}_{\text{noise interaction}}
\\[2pt]
&=\;\eta\Big(B_i+\mu_i^{\!\top}\sum_{j=1}^n \epsilon_j\Big).
\end{align*}
with $P=R^\top$ and $\mu_i=\mathbb{E}_x[Rg_i]$. A high-probability norm bound for $\sum_{j=1}^n\epsilon_j\sim\mathcal{N}(0,n\sigma^2I_d)$ (Lemma \ref{lemma2}) controls the worst-case angle between $\mu_i$ and the noise, producing the one-sided inequality and the factors $\|\mu_i\|$ and $\sqrt{d}$.
\item Scope of the probability statement
The probability $1-\delta$ is over the injected Gaussian noises; the expectation $\mathbb{E}_x$ is over client $i$’s local data. Thus the result guarantees a \emph{one-step}, \emph{first-order} expected decrease for a typical draw of noise, not global convergence.
\item Practical guidance
\begin{itemize}
    \item \textbf{Set noise conservatively}: choose $\sigma\le \kappa\,\sigma_{\mathrm{crit}}$ for some safety factor $\kappa\in(0,1)$.
    \item \textbf{Choose coordinates to noise}: prefer coordinates with small $|\mu_i|$ (low expected local gradient), reserving high-signal coordinates for little or no noise.
    \item \textbf{Aggregation rule}: averaging (instead of summing) reduces the aggregate noise scale from $\sigma\sqrt{n}$ to $\sigma/\sqrt{n}$, relaxing the per-client constraint.
    \item \textbf{Monitor \& adapt}: track empirical $\widehat{B}_i$ and $\widehat{\mu}_i$ to adapt $\sigma$ and $d$ across rounds and clients.
\end{itemize}
\item Limitations and assumptions
The guarantee is single-step and first-order; large steps or high curvature can violate the approximation. The analysis assumes zero-mean, isotropic, independent Gaussian noise on the selected coordinates and independence from the data. If $B_i\le 0$ (anti-alignment), ascensive steps may occur even for $\sigma=0$, reflecting data heterogeneity rather than privacy noise.
\end{enumerate}

\section{Experimental Results and Discussion} \label{Experiments}
\subsection{Experimental Setup}
\textbf{Dataset:} For our experiments, we used image classification as the core task for federated learning, employing the CIFAR-100 dataset for all evaluations. This benchmark dataset contains 60,000 32x32 color images across 100 classes, with 500 training images and 100 testing images per class.

\textbf{Models:} We conducted experiments using LeNet (88,648 parameters), MobileNetV3 (1,620,356 parameters), and ResNet-18 (11,227,812 parameters).

\textbf{Settings for Federated Learning:}
Our federated learning system consists of 3 clients followed the settings in~\cite{b4, b5}.

\textbf{Platform:}
All experiments were conducted on a machine equipped with a 12th Gen Intel Core™ i7-12700KF CPU (up to 5 GHz), 64 GB of memory, an NVIDIA RTX 4080 GPU, running on Ubuntu 20.04.

\FloatBarrier
\begin{figure}[h!]
    \centering
    \includegraphics[scale=0.6]{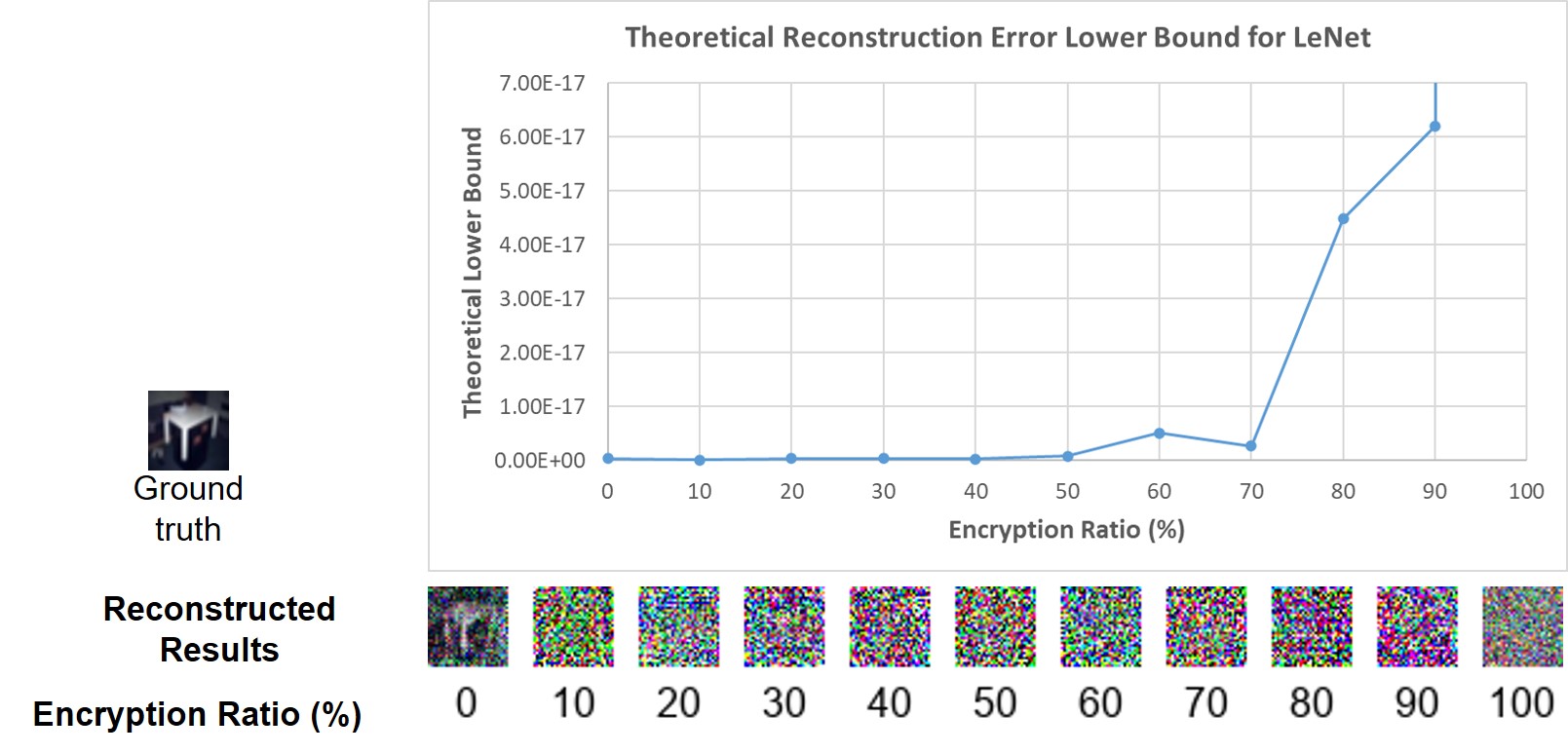}
    \caption{Experimental analysis and theoretical validation of the effect of encryption ratio for the LeNet architecture.}
    \vspace*{-5pt}
    \label{fig:encryption_ratio_effect_lenet}
\end{figure}

\begin{figure}[h!]
    \centering
    \includegraphics[scale=0.6]{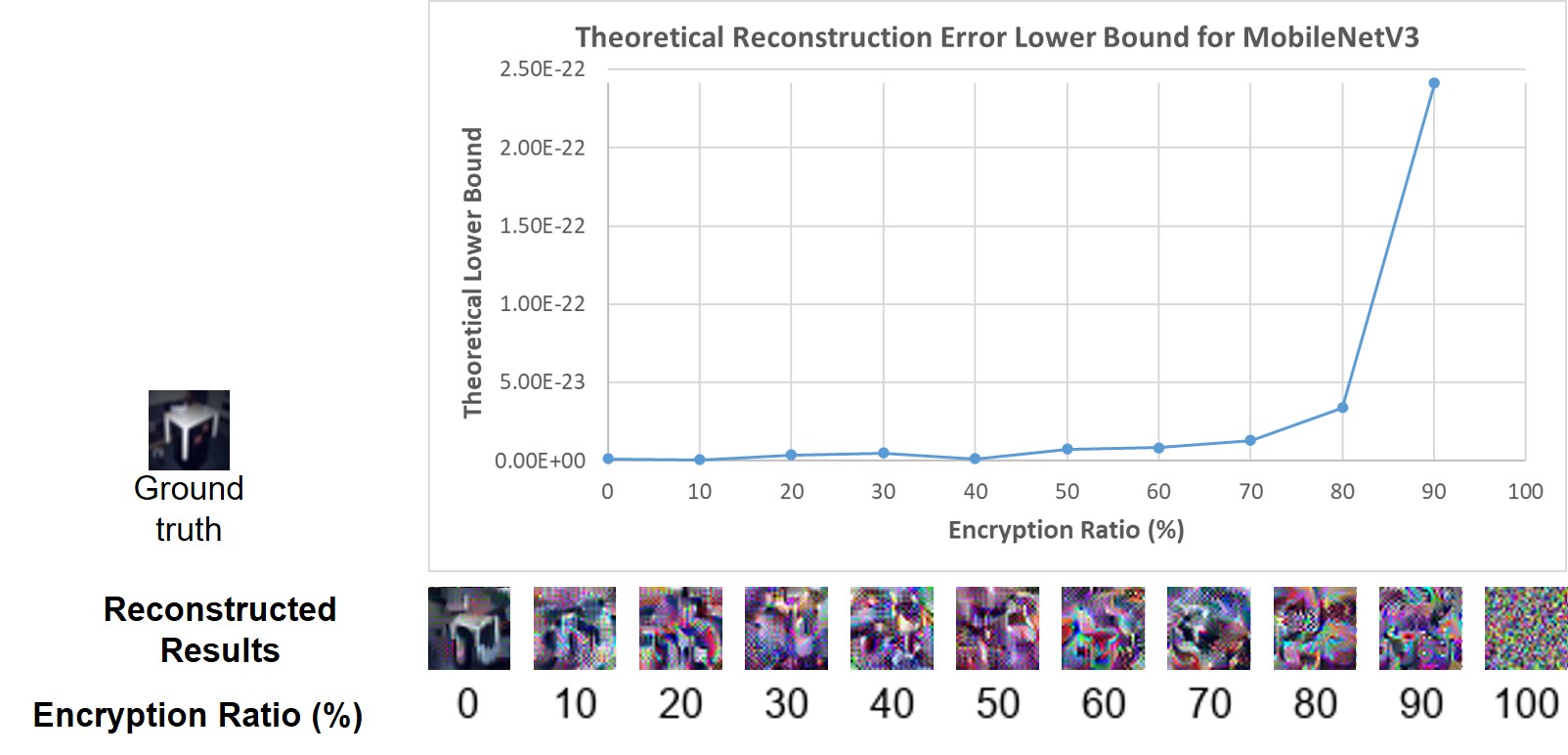}
    \caption{Experimental analysis and theoretical validation of the effect of encryption ratio for the MobilenetV3 architecture.}
    \vspace*{-5pt}
    \label{fig:encryption_ratio_effect_mobilenetv3}
\end{figure}

\begin{figure}[h!]
    \centering
    \includegraphics[scale=0.6]{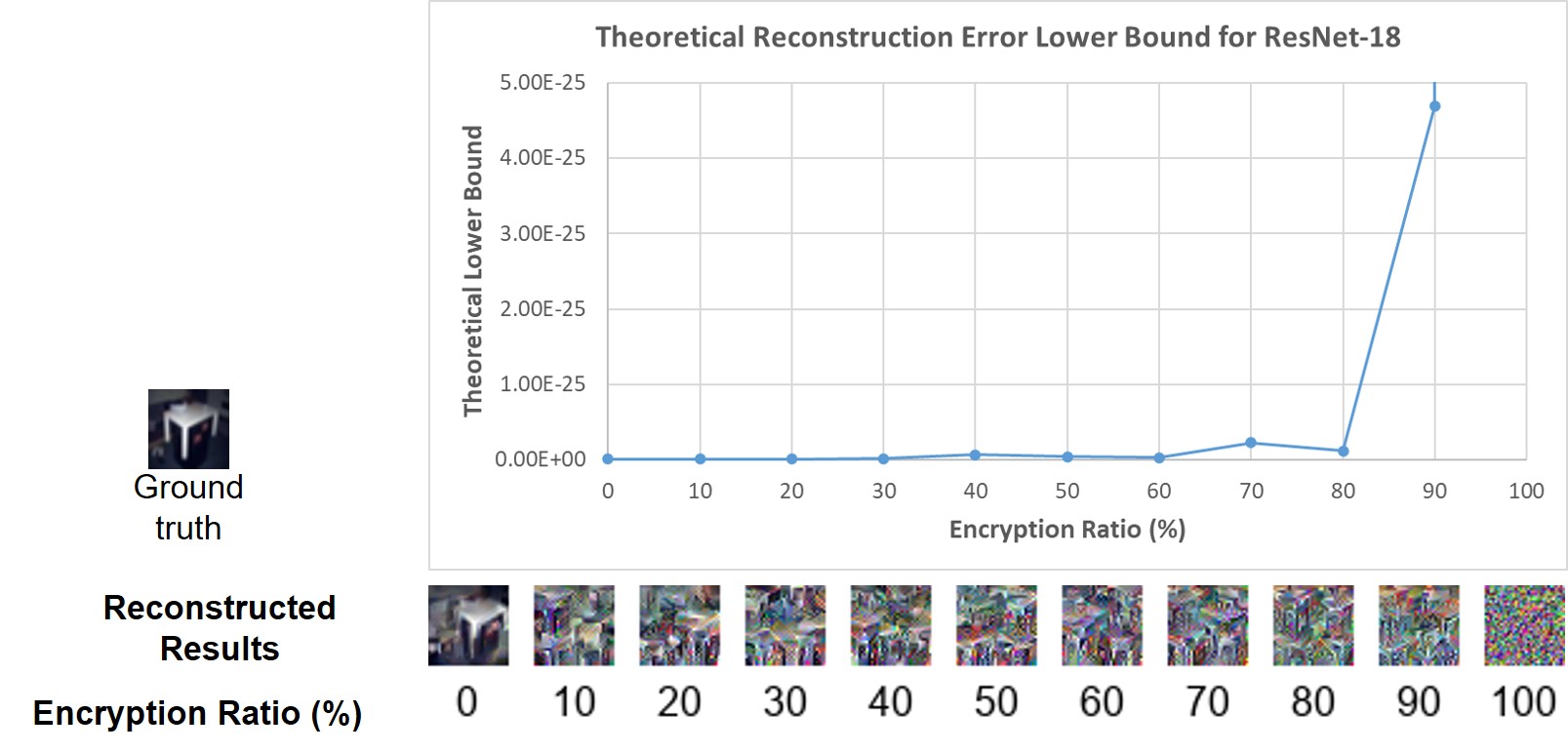}
    \caption{Experimental analysis and theoretical validation of the effect of encryption ratio for the ResNet-18 architecture.}
    \vspace*{-5pt}
    \label{fig:encryption_ratio_effect_resnet18}
\end{figure}
\FloatBarrier

\subsection{Factors Influencing Defense Against Model Inversion Attacks}

In this section, we analyze the key factors that influence the lower bound of the reconstruction error. Recall from Theoream \ref{theo:theorem1} that defines the bounding for reconstruction error,  Equation \ref{Equation_6} , we identified three major factors that govern the reconstruction error bound: the encryption ratio $z$, model complexity $D$, and gradient exposure $\mathbb{E}_{x \sim X} \| R \nabla_x g(x) \|^2_{\max}$. Here, we examine the impact of each factor through both theoretical analysis and empirical validation.

Figures \ref{fig:encryption_ratio_effect_lenet} to \ref{fig:encryption_ratio_effect_resnet18} illustrate the relationship between encryption ratio and reconstruction error for the three models. Theoretical curves show that the reconstruction error increases proportionally with the encryption ratio, a trend confirmed by our empirical results. This demonstrates that higher encryption ratios offer stronger protection against model inversion attacks.

The next factor influencing the reconstruction error is model size. Figure \ref{fig:model_complexity_effect} illustrates the impact of model complexity on the reconstruction error. Our theoretical analysis indicates that LeNet (88,648 parameters) exhibits a higher reconstruction error compared to ResNet-18 (11,227,812 parameters). This inverse relationship between model complexity and reconstruction error is confirmed by empirical results using Mean Squared Error (MSE), suggesting that simpler models may inherently provide stronger privacy guarantees against reconstruction attacks.

\begin{figure}[!htbp]
    \centering
    \includegraphics[scale=0.6]{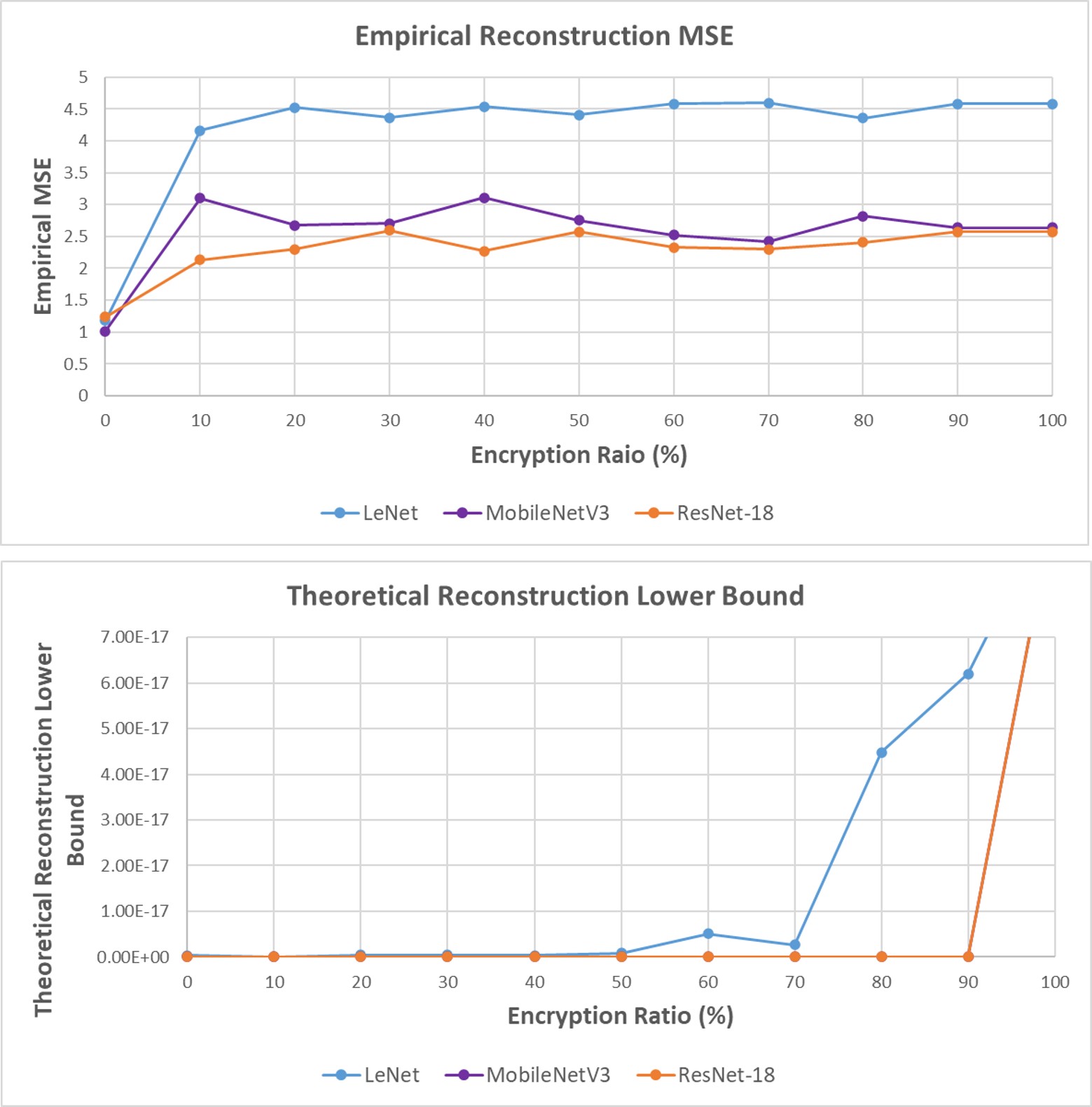}
    \caption{Effect of model complexity for reconstruction error. The theoretical reconstruction lower bounds for MobileNetV3 and ResNet-18 remain low up to a 90\% encryption ratio, appearing nearly overlapping in the visualization.}
    \vspace*{-5pt}
    \label{fig:model_complexity_effect}
\end{figure}

The third factor, gradient exposure, is illustrated in Figures \ref{fig:gradient_exposure_effect_lenet} through \ref{fig:gradient_exposure_effect_resnet18}. As the encryption ratio increases, gradient exposure decreases significantly, limiting the useful information available to attackers for reconstruction. This reduction in accessible gradient information directly contributes to higher reconstruction error, thereby enhancing defenses against model inversion attacks.

\begin{figure}[!htbp]
    \centering
    \includegraphics[scale=0.4]{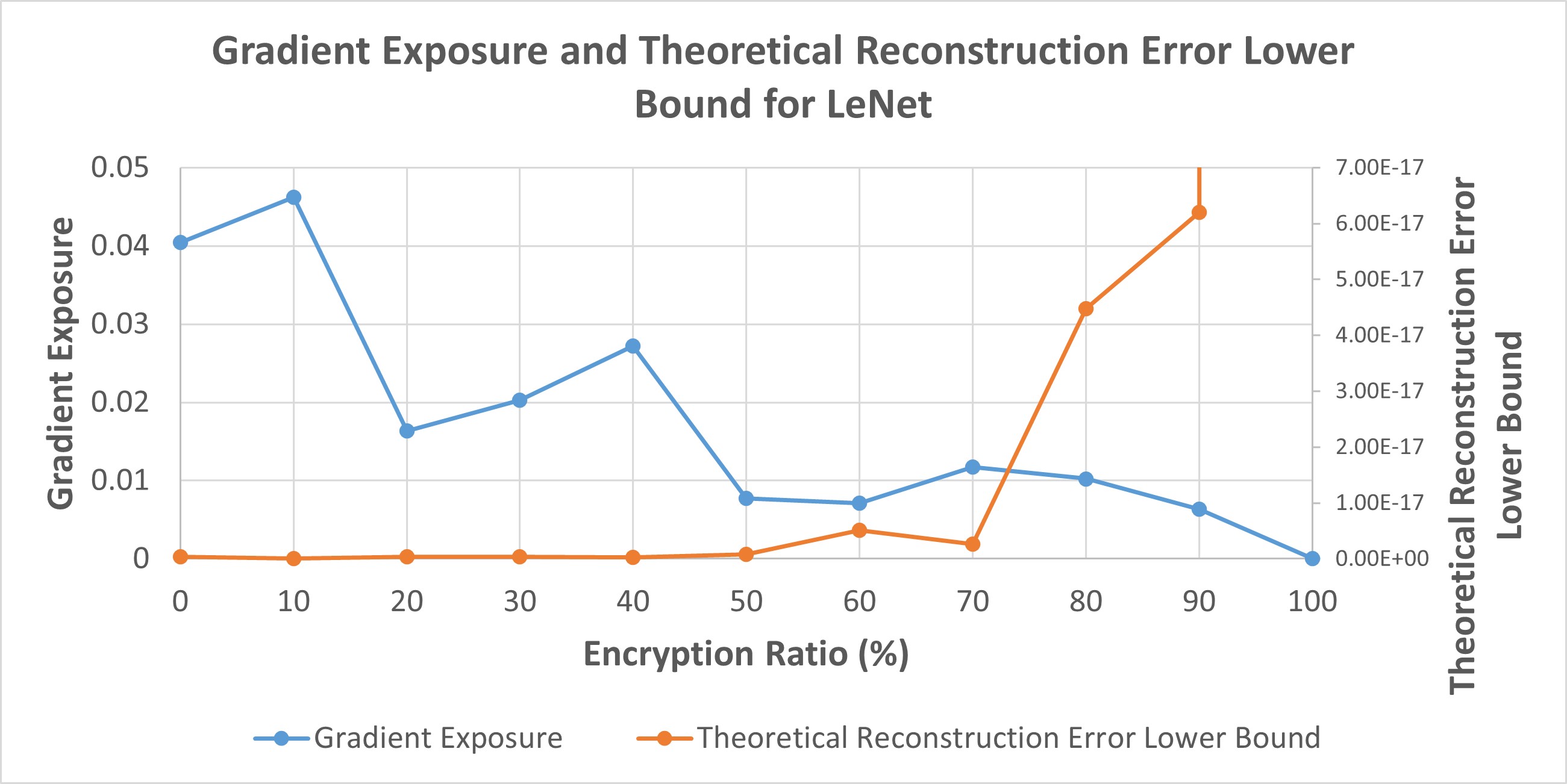}
    \caption{Effect of gradient exposure on theoretical reconstruction error lower bound for LeNet.}
    \vspace*{-5pt}
    \label{fig:gradient_exposure_effect_lenet}
\end{figure}

\begin{figure}[!htbp]
    \centering
    \includegraphics[scale=0.4]{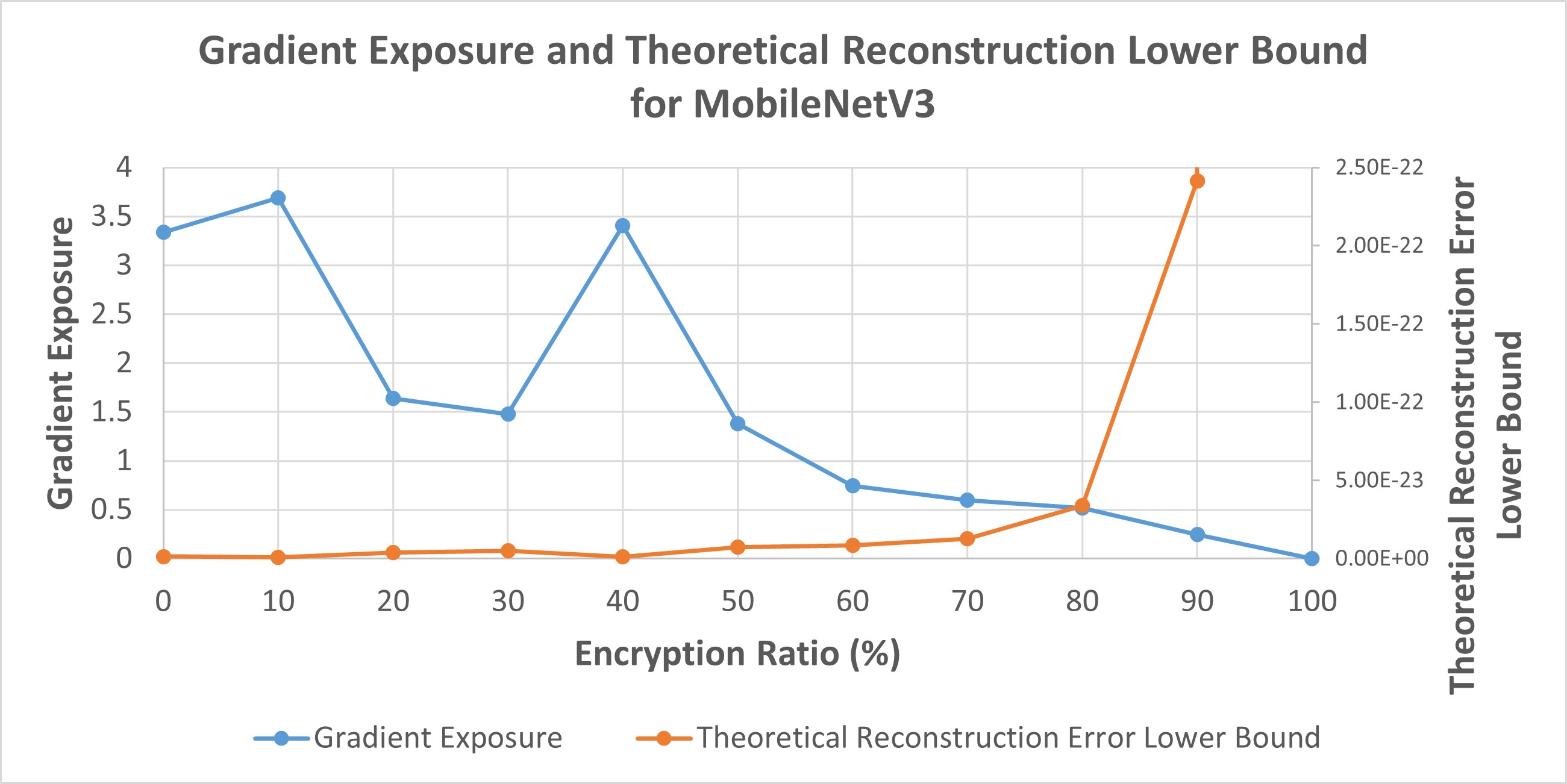}
    \caption{Effect of gradient exposure on theoretical reconstruction error lower bound for MobilenetV3.}
    \vspace*{-5pt}
    \label{fig:gradient_exposure_effect_mobilenetv3}
\end{figure}

\begin{figure}[!htbp]
    \centering
    \includegraphics[scale=0.4]{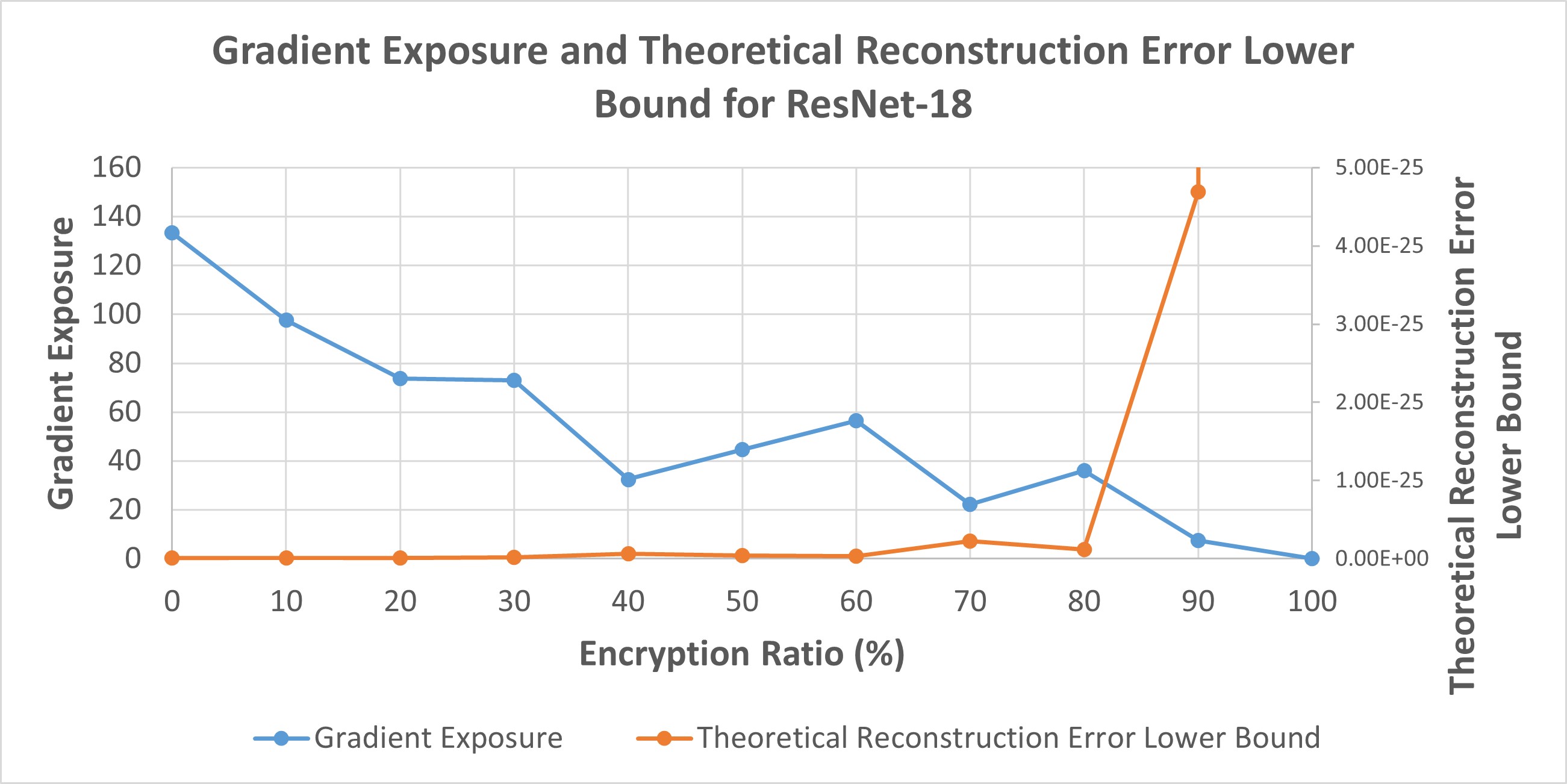}
    \caption{Effect of gradient exposure on theoretical reconstruction error lower bound for ResNet-18.}
    \vspace*{-5pt}
    \label{fig:gradient_exposure_effect_resnet18}
\end{figure}

%=============================================
%Dumindu Edited until this point!
%=============================================

\subsection{Determining the Upper Bound of Additive Noise}
In addition to the findings above, this section investigates the upper bound of additive noise that preserves training utility by analyzing loss convergence under varying noise levels. Training performance is evaluated by measuring the loss of the global model, aggregated from client updates, on the training dataset. The additive noise values are set to $\sigma = 0$, $\sigma = 10^{-6}$, $\sigma = 10^{-3}$, $\sigma = 1$, and $\sigma = 10$, with an encryption ratio of 0\%, to observe the effect of noise across all gradients.

\subsubsection{Empirical Findings}
Our experiments reveal three distinct convergence behaviors under different noise settings, as illustrated in Figures \ref{fig:loss_convergence_noise_lenet} through \ref{fig:loss_convergence_noise_resnet18}:

\begin{description}
  \item[\textbf{Low noise} ($\sigma = 10^{-6}$):] Training convergence closely approximates the noise-free baseline, indicating minimal impact on model performance.
  \item[\textbf{Moderate noise} ($\sigma = 10^{-3}$):] Training convergence begins to degrade, resulting in higher final loss values compared to the baseline after the training duration.
  \item[\textbf{High noise} ($\sigma \geq 1$):] Training fails to converge entirely, with loss values remaining consistently high throughout all training rounds.
\end{description}

\subsubsection{Upper Bound for Additive Noise based on Fixed Noise}
We evaluate the final training loss under varying noise settings and calculate the loss standard deviation relative to the baseline case (noise std = 0) to assess the deviation in training behavior under noise perturbations. As shown in Table \ref{tab:final_loss} and Table \ref{tab:model_performance}, when the noise standard deviation reaches $10^{-3}$ or higher, training convergence is significantly degraded, with models failing to converge entirely at noise levels of 1 or 10.

Finally, from these empirical results, we identify $10^{-6}$ as the practical upper bound for additive noise. This threshold is chosen because it satisfies the following critical conditions:

\begin{itemize}
    \item Training utility remains substantially unaffected
    \item The noise level remains consistent with the theoretical framework constraints
    \item Convergence behavior closely resembles the noise-free case
    \item It satisfies the mathematical requirements of the BCRLB while preserving training utility
\end{itemize}

% The empirically derived upper bound of additive noise provides a concrete parameter for the proposed mathematical framework. In particular, it enables us to evaluate Model Inversion Attack reconstruction error under selective encryption defenses without compromising training performance. Thus, the noise upper bound acts as a critical bridge between theoretical privacy analysis and practical federated learning systems.

\begin{figure}[!htbp]
    \centering
    \includegraphics[scale=0.25]{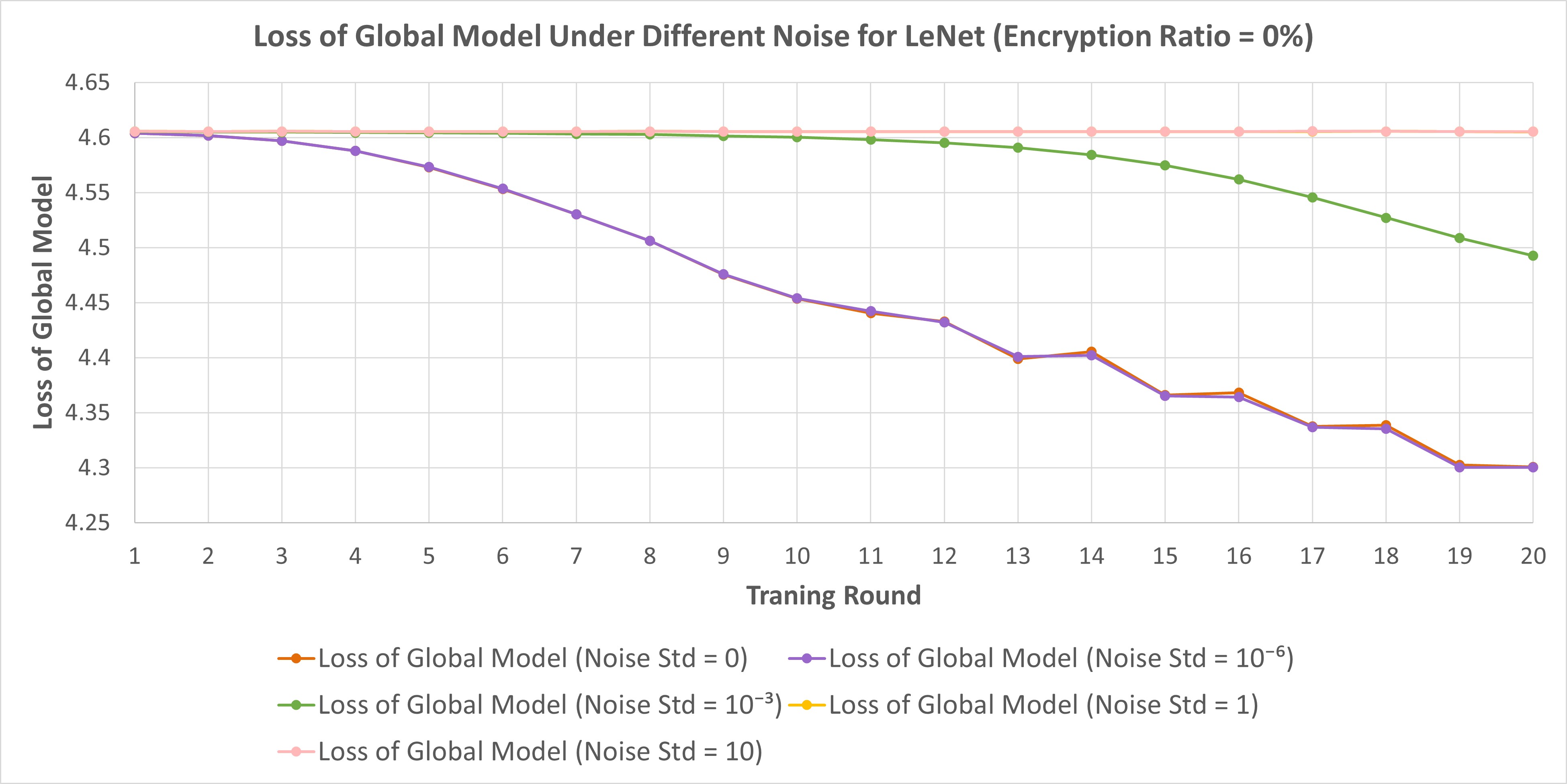}
    \caption{Loss convergence of the global model for LeNet.}
    \vspace*{-5pt}
    \label{fig:loss_convergence_noise_lenet}
\end{figure}

\begin{figure}[!htbp]
    \centering
    \includegraphics[scale=0.25]{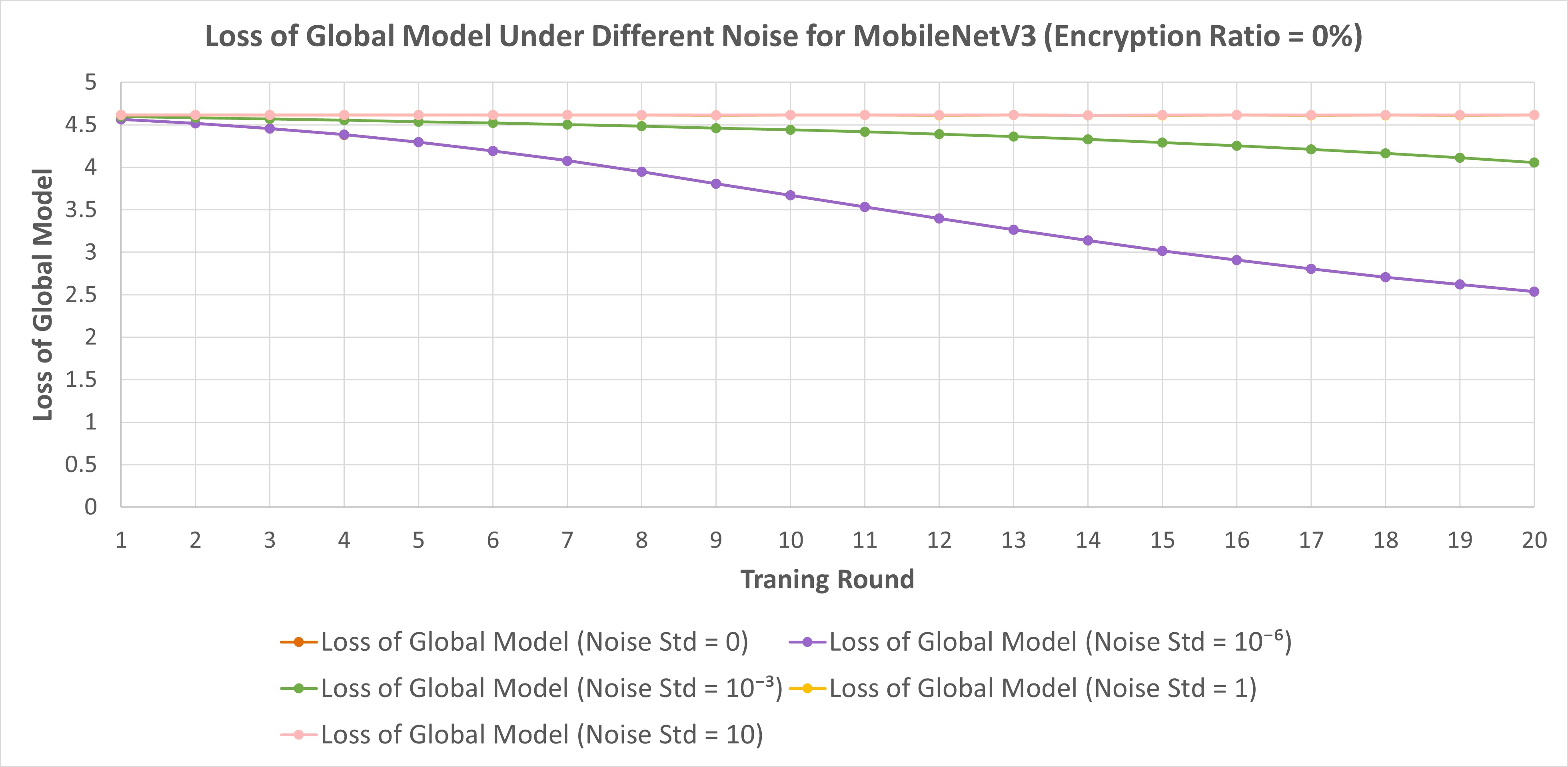}
    \caption{Loss convergence of the global model for MobileNetV3.}
    \vspace*{-5pt}
    \label{fig:loss_convergence_noise_mobilenetv3}
\end{figure}

\begin{figure}[!htbp]
    \centering
    \includegraphics[scale=0.25]{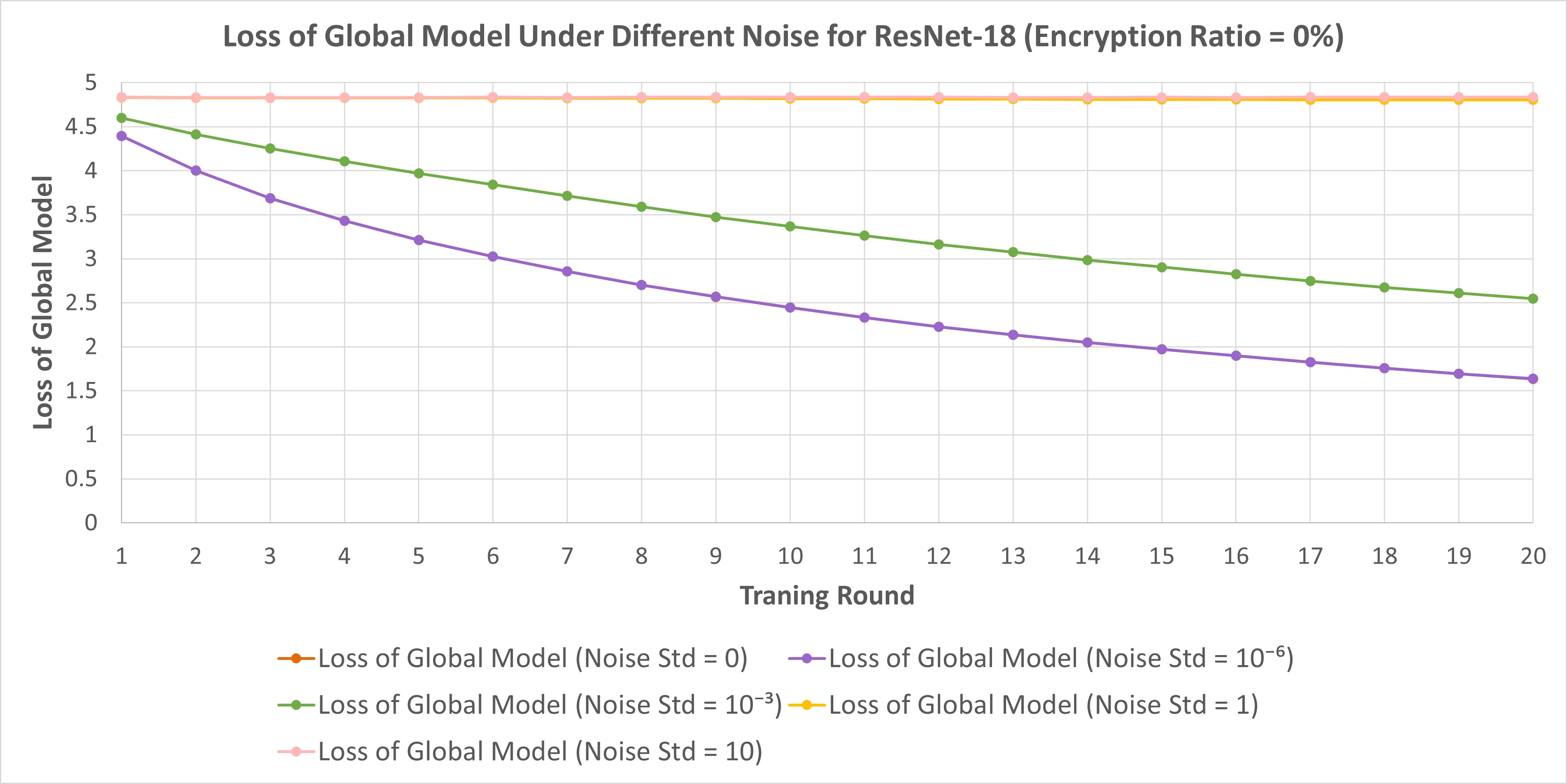}
    \caption{Loss convergence of the global model for ResNet-18.}
    \vspace*{-5pt}
    \label{fig:loss_convergence_noise_resnet18}
\end{figure}

\begin{table}[h!]
\centering
\caption{Final loss values across different noise settings.}
\label{tab:final_loss}
\begin{tabular}{|l|l|c|c|c|c|c|}
\hline
\multicolumn{2}{|c|}{\multirow{2}{*}{\textbf{Final Loss}}} & \multicolumn{5}{c|}{\textbf{Noise std}} \\
\cline{3-7}
\multicolumn{2}{|c|}{} & \textbf{0} & \textbf{$10^{-6}$} & \textbf{$10^{-3}$} & \textbf{1} & \textbf{10} \\
\hline
\multirow{3}{*}{\textbf{Model}} & \textbf{LeNet} & 4.300 & 4.300 & 4.493 & 4.606 & 4.606 \\
\cline{2-7}
& \textbf{MobileNetV3} & 2.537 & 2.536 & 4.057 & 4.613 & 4.614 \\
\cline{2-7}
& \textbf{ResNet-18} & 1.637 & 1.637 & 2.546 & 4.804 & 4.836 \\
\hline
\end{tabular}
\end{table}

\begin{table}[h!]
\centering
\caption{Loss standard deviation to baseline across different noise settings.}
\label{tab:model_performance}
\begin{tabular}{|l|l|c|c|c|c|c|}
\hline
\multicolumn{2}{|c|}{\multirow{2}{*}{\textbf{Loss std to Baseline}}} & \multicolumn{5}{c|}{\textbf{Noise std}} \\
\cline{3-7}
\multicolumn{2}{|c|}{} & \textbf{0} & \textbf{$10^{-6}$} & \textbf{$10^{-3}$} & \textbf{1} & \textbf{10} \\
\hline
\multirow{3}{*}{\textbf{Model}} & \textbf{LeNet} & 0 & 0.0015 & 0.0789 & 0.1010 & 0.1010 \\
\cline{2-7}
& \textbf{MobileNetV3} & 0 & 0.0004 & 0.5174 & 0.6613 & 0.6616 \\
\cline{2-7}
& \textbf{ResNet-18} & 0 & 0.0001 & 0.2064 & 0.7884 & 0.7974 \\
\hline
\end{tabular}
\end{table}

\begin{figure}[!htbp]
    \centering
    \includegraphics[scale=0.3]{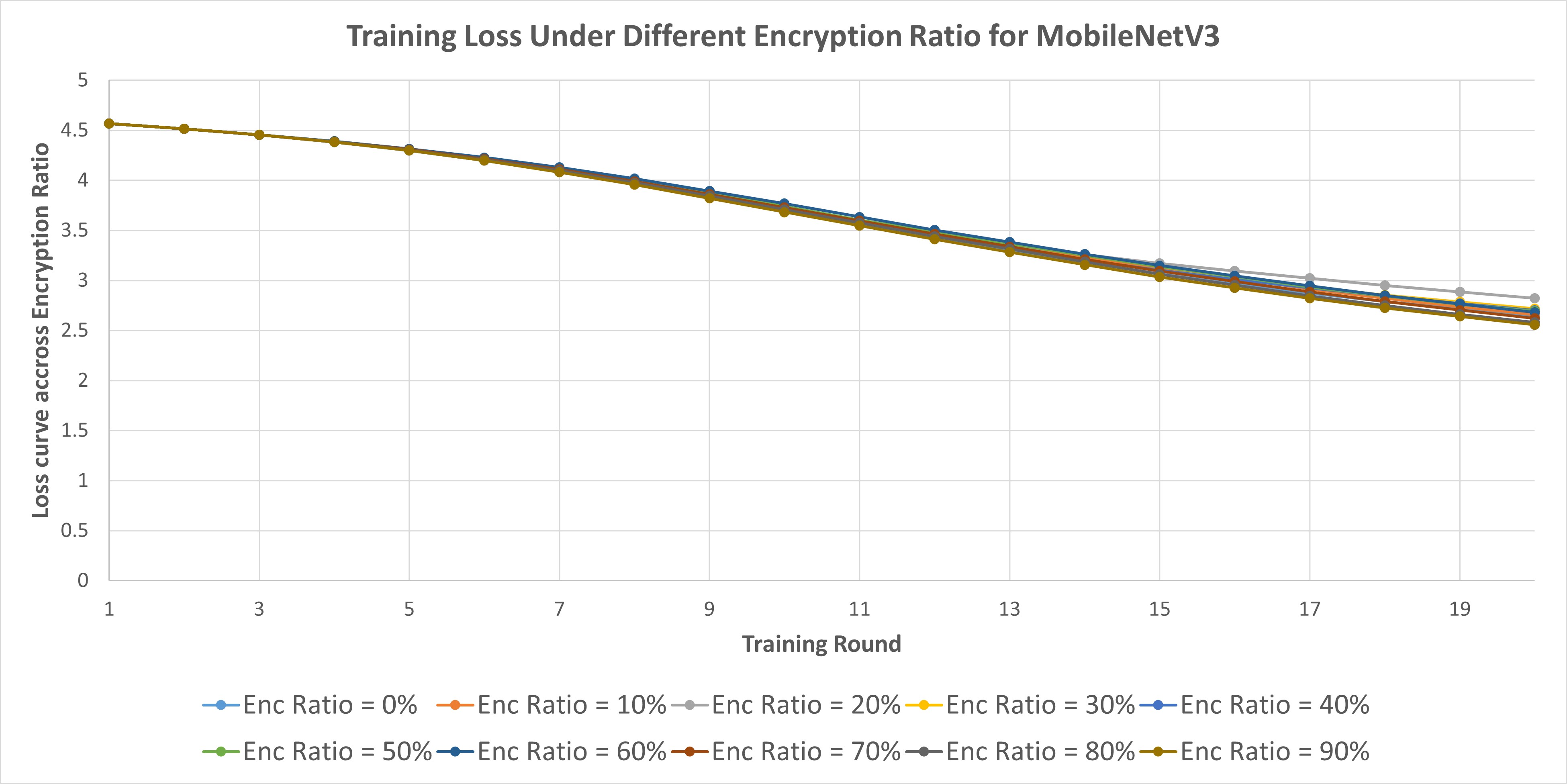}
    \caption{The training utility is not affected under different encryption ratio. The noise is initialized as $10^{-6}$ and changed adaptively over training rounds.}
    \vspace*{-5pt}
    \label{fig:loss_encryption_ratio}
\end{figure}

\begin{figure}[!htbp]
    \centering
    \includegraphics[scale=0.5]{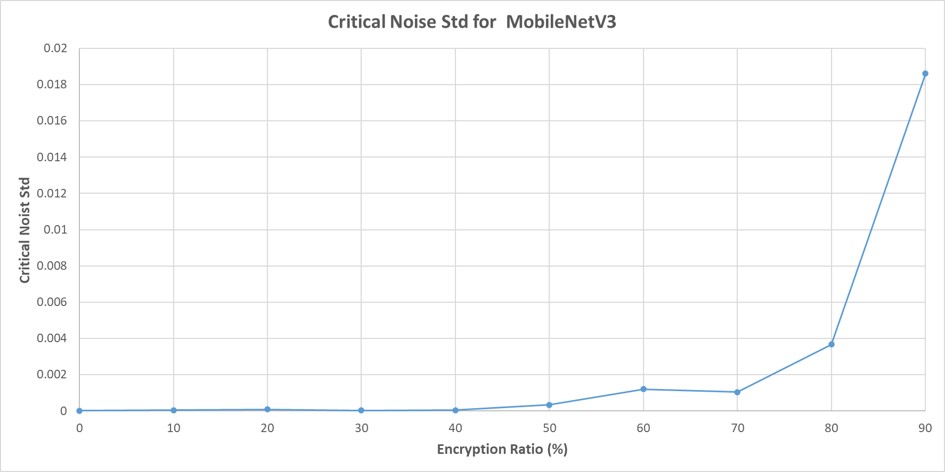}
    \caption{Critical noise under different encryption ratio for MobileNetV3.}
    \vspace*{-5pt}
    \label{fig:dynamic_critial_noise_mobilenetv3}
\end{figure}

\subsubsection{Effect of Encryption Ratio on Critical Noise}
The previous section established an upper bound for fixed additive noise. However, manual parameter tuning requires developers to explore various combinations to determine optimal noise levels. Building on the utility analysis in Section \ref{utility_reduction}, we now evaluate the critical additive noise under the selective encryption mechanism, which enables dynamic optimization for noise while preserving both model utility and privacy protection for unencrypted gradients. Equation \ref{Equation_critical_noise} proposes calculating critical noise for each training round based on the gradients available to potential attackers. First, training utility is not compromised under different encryption ratios with adaptive noise, as shown in Figure \ref{fig:loss_encryption_ratio}. Second, Figure \ref{fig:dynamic_critial_noise_mobilenetv3} illustrates how critical noise varies with encryption ratios for MobileNetV3. The results demonstrate a positive correlation between encryption ratio and critical noise levels. Higher encryption ratios leave fewer gradients unencrypted, thereby permitting increased noise injection without compromising model performance.

% \begin{table}[h!]
% \centering
% \caption{Loss standard deviation to baseline across different noise settings.}
% \label{tab:model_performance}
% \begin{tabular}{|l|l|c|c|c|c|c|}
% \hline
% \multicolumn{2}{|c|}{\textbf{Loss std to Baseline}} & \multicolumn{5}{c|}{\textbf{Noise std}} \\
% \hline
% \multicolumn{2}{|c|}{} & \textbf{0} & \textbf{1e-6} & \textbf{1e-3} & \textbf{1} & \textbf{10} \\
% \hline
% \multirow{3}{*}{\textbf{Model}} & \textbf{LeNet} & 0 & 0.0015 & 0.0789 & 0.1010 & 0.1010 \\
% \cline{2-7}
% & \textbf{MobileNetV3} & 0 & 0.0004 & 0.5174 & 0.6613 & 0.6616 \\
% \cline{2-7}
% & \textbf{ResNet-18} & 0 & 0.0001 & 0.2064 & 0.7884 & 0.7974 \\
% \hline
% \end{tabular}
% \end{table}

\needspace{10\baselineskip}
\section{Conclusion} \label{Conclusion}
In this paper, we presented a theoretical framework for analyzing model inversion attacks under selective encryption defense mechanisms. Our analysis establishes a lower bound on attack effectiveness and highlights key factors—including encryption ratio, model complexity, and gradient exposure—that govern an adversary’s ability to reconstruct private training data from shared gradients. To complement the theory, we conducted an empirical study that examines how these factors influence defense robustness. In particular, we analyzed the role of noise added to unencrypted gradients for deriving the Bayesian Cramér-Rao lower bound, identifying the minimal noise threshold required to balance model utility and privacy. These findings provide both theoretical and practical insights into the design of efficient selective encryption strategies for federated learning. They clarify how selective encryption mitigates model inversion attacks and offer concrete guidance for parameter selection in real-world applications. Looking ahead, we will extend this framework to language models in federated learning, beginning with small-scale models and progressively scaling to larger architectures. Our future work will also explore hybrid homomorphic encryption–based selective encryption methods and investigate quantum-resilient approaches to strengthen long-term privacy guarantees.

% \begin{figure*}[htb]
% \centering
% \begin{multicols}{2}
%     \centering
%     \includegraphics[scale=0.3]{utility_reduction_mean_LeNet}
%     \captionof{figure}{Effect of Noise on Loss Reduction for LeNet}
%     \label{fig:utility_reduction_mean_LeNet}

%     \includegraphics[scale=0.3]{utility_reduction_mean_mobilenetv3}
%     \captionof{figure}{Effect of Noise on Loss Reduction for MobileNetV3}
%     \label{fig:utility_reduction_mean_MobileNet}
    
%     \includegraphics[scale=0.3]{utility_reduction_mean_ResNet18}
%     \captionof{figure}{Effect of Noise on Loss Reduction for ResNet-18}
%     \label{fig:utility_reduction_mean_ResNet18}
    
%     \columnbreak

%     \includegraphics[scale=0.3]{loss_convergence_noise_lenet}
%     \captionof{figure}{Loss convergence of the global model for LeNet}
%     \label{fig:loss_convergence_noise_lenet}
    
%     \includegraphics[scale=0.3]{loss_convergence_noise_mobilenetv3}
%     \captionof{figure}{Loss convergence of the global model for MobileNetV3}
%     \label{fig:loss_convergence_noise_mobilenetv3}

%     \includegraphics[scale=0.3]{loss_convergence_noise_resnet18}
%     \captionof{figure}{Loss convergence of the global model for MobileNetV3}
%     \label{fig:loss_convergence_noise_resnet18}

% \end{multicols}
% \end{figure*}

\onecolumn
\appendix

\section{MATHEMATICAL DERIVATIONS}

Here we provide the proof for Theorem 1 and 2 along with necessary Lemmas.

\begin{lemma}
    Let $J_F(x)$ be the Fisher information matrix as defined in (\ref{eq:bcrlb_JF}). With $u \sim \mathcal{N}(Rg(x), \sigma^2I)$ being the observed noisy gradients, we have:
\begin{equation}
    \label{lemma1}
    J_F(x) = \frac{1}{\sigma^2} \nabla_x g(x)^T R^T R\nabla_x g(x) \in \mathbb{R}^{m \times m} 
\end{equation}
\end{lemma}
\begin{proof}
    Starting with the definition of $J_F(x)$ from (\ref{eq:bcrlb_JF}), i.e., $
J_F(x) = \mathbb{E}_{u|x}\left[\nabla_x \log p(u|x) \nabla_x \log p(u|x)^T\right]$, we compute the gradient of the log-likelihood function:
\begin{align*}
    \nabla_x \log p(u|x) = -\frac{1}{2} \nabla_x ((u - R g(x))^T (\sigma^{2} I)^{-1} (u - R g(x)))
\end{align*}
where $$(\sigma^2 I) = \Sigma$$
Expanding the expression:
\begin{align*}
    \nabla_x \log p(u | x) &= -\frac{1}{2} \nabla_x \left((u - R g(x))^T \Sigma^{-1} (u - R g(x))\right) \\
    &= -\frac{1}{2} \nabla_x \left( u^T \Sigma^{-1} u - 2 g(x)^T R^T \Sigma^{-1} u + g(x)^T R^T \Sigma^{-1} R g(x) \right) \\
    &= -\frac{1}{2} \left( 0 - 2 \nabla_x g(x)^T R^T \Sigma^{-1} u + 2 \nabla_x g(x)^T R^T \Sigma^{-1} R g(x) \right) \\
    &= \nabla_x g(x)^T R^T \Sigma^{-1} u - \nabla_x g(x)^T R^T \Sigma^{-1} R g(x) \\
    &= \nabla_x g(x)^T R^T \Sigma^{-1} (u - Rg(x))
\end{align*}
Apply the gradient of the log-likelihood function into the Fisher Information matrix:
\begin{align*}
    \mathbb{E}_{u|x} \left[ \nabla_x \log p(u|x) \cdot \nabla_x \log p(u|x)^T \right] &= \mathbb{E}_{u|x} \left[ (\nabla_x g(x)^T R^T \Sigma^{-1} (u - R g(x)) \cdot ((u - R g(x))^T \Sigma^{-1} R \nabla_x g(x)) \right]
\end{align*}
Since:
\begin{align*}
    \mathbb{E}_{u|x} \left[ (u - R g(x))(u - R g(x))^T \right] = \Sigma
\end{align*}
we obtain:
\begin{align*}
    \mathbb{E}_{u|x} \left[ \nabla_x \log p(u|x) \cdot \nabla_x \log p(u|x)^T \right] &= \nabla_x g(x)^T R^T \Sigma^{-1} \mathbb{E}_{u|x} \left[ (u - R_g(x))(u - R_g(x))^T \right] \Sigma^{-1} R \nabla_x g(x) \\
    &= \nabla_x g(x)^T R^T \Sigma^{-1} \Sigma \Sigma^{-1} R \nabla_x g(x) \\
    &= \nabla_x g(x)^T R^T \Sigma^{-1} R \nabla_x g(x)
\end{align*}
Thus, the Fisher Information matrix is:
\begin{align*}
    J_F(x) = \frac{1}{\sigma^2} \nabla_x g(x)^T R^T R \nabla_x g(x) \quad \in \mathbb{R}^{m \times m}
\end{align*}
\end{proof}

\subsection{\textbf{Proof of Theorem 1}} \label{derive_reconstruction_lower_bound}
Here we start with the BCRLB based reconstruction error lower bound developed in Chen et al.~\cite{b34}. It should be noted that owing to an inconsistency in their proof, our equation (\ref{app1}) looks different Theorem 1 in \cite{b34}. Hence, with $E_A$ representing the reconstruction error as defined in (\ref{eq:reconstruction_error}), and with the definition of $J_B$ as defined in (\ref{Jb}) and its relationship with reconstruction error (\ref{eq:bcrlb_main}) we have :
\begin{equation}
    E_A \geq tr(J_B^{-1}) \geq \frac{m^2}{\mathbb{E}_{x\sim \mathcal{X}}[\text{tr}(J_F(x))] + m \cdot \lambda_1(J_P)}
    \label{app1}
\end{equation}

Expanding the trace term we get:
\begin{equation*}
\begin{split}
    \mathbb{E}_{x\sim \mathcal{X}}[\text{tr}(J_F(x))] &= \mathbb{E}_{x\sim \mathcal{X}}\left[\text{tr}\left(\frac{1}{\sigma^2}\nabla_x g(x)^T R^T R\nabla_x g(x)\right)\right], \quad \text{(from Lemma 1)}\\
    &= \frac{1}{\sigma^2}\mathbb{E}_{x\sim \mathcal{X}}\left[\sum_{i=1}^{m}\|R\nabla_{x_i}g(x)\|_2^2\right]\\
    &= \frac{1}{\sigma^2}\mathbb{E}_{x\sim \mathcal{X}}\left[\sum_{i=1}^{m}\sum_{j=1}^{d}(R\nabla_{x_i}g(x))[j]^2\right]\\
    &= \frac{1}{\sigma^2}\mathbb{E}_{x \sim X} \left[\sum_{i=1}^{m} \sum_{j=1}^{d} (R \nabla_{x} g(x))[j, i]^2\right] \\
    &\leq \frac{md}{\sigma^2}\mathbb{E}_{x\sim \mathcal{X}}\|R\nabla_x g(x)\|^2_{\max}
\end{split}
\end{equation*}

Assuming the gradient of the model gradient vector $g(x)$ with respect to data $x$ quantifies its sensitivity with respect to each data feature $x_i$, we obtain the last inequality by extracting the max sensitivity of an unencrypted gradient coordinate with respect to a data feature. By noting, the relationship between unencrypted parameters $d$ and total parameters $D$ with encryption ratio $z$ we obtain:
\begin{equation*}
    z = \frac{D-d}{D} \Rightarrow d = D(1-z)
\end{equation*}

Substituting these results in (\ref{app1}) we get:
\begin{equation*}
\begin{split}
    E_A &\geq \frac{m^2}{\mathbb{E}_{x\sim D}[\text{tr}(J_F(x))] + m \cdot \lambda_1(J_P)}\\
    &\geq \frac{m^2}{\frac{md}{\sigma^2}\mathbb{E}_{x\sim D}\|R\nabla_x g(x)\|^2_{\max} + m \cdot \lambda_1(J_P)}\\
    &= \frac{m}{\frac{d}{\sigma^2}\mathbb{E}_{x\sim D}\|R\nabla_x g(x)\|^2_{\max} + \lambda_1(J_P)}\\
    &= \frac{m}{\frac{D(1-z)}{\sigma^2}\mathbb{E}_{x\sim D}\|R\nabla_x g(x)\|^2_{\max} + \lambda_1(J_P)}
\end{split}
\end{equation*}
This final expression shows how the encryption ratio $z$, model parameters $D$, noise variance $\sigma^2$ directly influences the reconstruction error lower bound.

\hfill $\square$

\vspace{0.1in}

\begin{lemma}[Gaussian norm concentration for summed client noise]
\label{lemma2}
Let $\epsilon_1,\ldots,\epsilon_n$ be i.i.d.\ $\mathcal N(0,\sigma^2 I_d)$ and set $S:=\sum_{i=1}^n \epsilon_i$. For any $\delta\in(0,1)$,
\[
\Pr\!\Big(\,\|S\| \;\ge\; \sigma\sqrt{n}\,\big(\sqrt{d}+\sqrt{2\ln(1/\delta)}\big)\Big)\;\le\;\delta.
\]
Equivalently, with probability at least $1-\delta$,
\[
\|S\|^2 \;\le\; \sigma^2 n\Big(d+2\sqrt{d\ln(1/\delta)}+2\ln(1/\delta)\Big).
\]
\end{lemma}
\begin{proof}
Since $\epsilon_i\sim \mathcal N(0,\sigma^2 I_d)$ are independent, $S\sim \mathcal N(0,n\sigma^2 I_d)$. Normalize $Z:=S/(\sigma\sqrt{n})\sim \mathcal N(0,I_d)$, so $\|Z\|^2\sim \chi^2_d$. The Laurent--Massart upper tail for $\chi^2_d$ (special case of Proposition 1 in \cite{b35}) states that for all $t\ge 0$,
\[
\Pr\!\big(\|Z\|^2 > d+2\sqrt{dt}+2t\big)\le e^{-t}.
\]
Taking $t=\ln(1/\delta)$ yields, with probability at least $1-\delta$,
\[
\|Z\|^2 \le d+2\sqrt{d\ln(1/\delta)}+2\ln(1/\delta).
\]
Rescaling back gives the stated squared--norm bound $\|S\|^2 \le \sigma^2 n\big(d+2\sqrt{d\ln(1/\delta)}+2\ln(1/\delta)\big)$. For a cleaner norm bound, take square roots and use $\sqrt{d+2\sqrt{du}+2u}\le \sqrt d+\sqrt{2u}$ with $u=\ln(1/\delta)$, which implies
\[
\|S\|\le \sigma\sqrt{n}\big(\sqrt d+\sqrt{2\ln(1/\delta)}\big)
\]
with probability at least $1-\delta$.
\end{proof}

\subsection{\textbf{Proof of Theorem 2}}
\label{derive_privacy_utility}
\begin{proof}
A first-order Taylor expansion at $\theta$ along $\theta^+=\theta-\eta Q(x)$ gives
\[
L_i(x,\theta^{+}) \approx L_i(x,\theta) - \eta\, g_i(x)^{\!\top} Q(x),
\]
hence, conditioning on the realized noises and taking expectation over $x$,
\[
\mathbb{E}_x\!\big[L_i(x,\theta)-L_i(x,\theta^{+})\big]
\;\approx\;
\eta\,\mathbb{E}_x\!\big[g_i(x)^{\!\top}G(x)\big]
+\eta\,\mathbb{E}_x[g_i(x)]^{\!\top}P\sum_{j=1}^{n}\epsilon_j.
\]
With $P=R^\top$ we have $\mathbb{E}_x[g_i(x)]^{\!\top}P=\big(\mathbb{E}_x[Rg_i(x)]\big)^{\!\top}=\mu_i^{\!\top}$, so
\[
\mathbb{E}_x\!\big[L_i(x,\theta)-L_i(x,\theta^{+})\big]
\;\approx\;
\eta\Big(B_i+\mu_i^{\!\top}\sum_{j=1}^{n}\epsilon_j\Big).
\]
By Cauchy--Schwarz,
\[\mu_i^{\!\top}\sum_{j=1}^{n}\epsilon_j \ge -\|\mu_i\|\,\big\|\sum_{j=1}^{n}\epsilon_j\big\|.\]
Since $\sum_{j=1}^{n}\epsilon_j\sim\mathcal N(0,n\sigma^2 I_d)$, using Lemma 2, for any $\delta\in(0,1)$, with probability at least $1-\delta$,
\[
\Big\|\sum_{j=1}^{n}\epsilon_j\Big\|
\;\le\;
\sigma\sqrt{n}\,\Big(\sqrt{d}+\sqrt{2\ln(1/\delta)}\Big).
\]
Combining, with probability at least $1-\delta$,
\[
\mathbb{E}_x\!\big[L_i(x,\theta)-L_i(x,\theta^{+})\big]
\;\gtrsim\;
\eta\Big(B_i-\sigma\sqrt{n}\,\|\mu_i\|\,(\sqrt{d}+\sqrt{2\ln(1/\delta)})\Big),
\]
which is nonnegative whenever
\[\sigma \le \dfrac{B_i}{\sqrt{n}\,\|\mu_i\|\,(\sqrt{d}+\sqrt{2\ln(1/\delta)})}\].

If $B_i\le 0$, the right-hand will can be non-positive. Hence, even at $\sigma=0$, ascensive steps can occur. This establishes the claim (up to first-order approximation error).
\end{proof}

\end{document}